\documentclass[a4paper,USenglish,cleveref, autoref, thm-restate]{lipics-v2021}
\usepackage{graphicx}
\usepackage{algorithm}
\usepackage[noend]{algorithmic}
\usepackage{amsfonts}
\usepackage{amsmath}
\title{Optimal Bounds for Weak Consistent Digital Rays in 2D} %TODO Please add

\author{Matt Gibson-Lopez}{Department of Computer Science, The University of Texas at San Antonio, San Antonio, USA }{matthew.gibson@utsa.edu}{}{}

\author{Serge Zamarripa}{Department of Computer Science, The University of Texas at San Antonio, San Antonio, USA }{sergio.zamarripa@my.utsa.edu}{}{}

\authorrunning{M. Gibson-Lopez and S. Zamarripa} %TODO mandatory. First: Use abbreviated first/middle names. Second (only in severe cases): Use first author plus 'et al.'

\Copyright{Matt Gibson-Lopez and Serge Zamarripa} %TODO mandatory, please use full first names. LIPIcs license is "CC-BY";  http://creativecommons.org/licenses/by/3.0/

\ccsdesc{Theory of computation~Computational geometry} %TODO mandatory: Please choose ACM 2012 classifications from https://dl.acm.org/ccs/ccs_flat.cfm 

\keywords{Digital Geometry, Consistent Digital Rays} %TODO mandatory; please add comma-separated list of keywords

\category{} %optional, e.g. invited paper

\relatedversion{} %optional, e.g. full version hosted on arXiv, HAL, or other respository/website
\nolinenumbers %uncomment to disable line numbering

\hideLIPIcs  %uncomment to remove references to LIPIcs series (logo, DOI, ...), e.g. when preparing a pre-final version to be uploaded to arXiv or another public repository

\EventEditors{John Q. Open and Joan R. Access}
\EventNoEds{2}
\EventLongTitle{42nd Conference on Very Important Topics (CVIT 2016)}
\EventShortTitle{CVIT 2016}
\EventAcronym{CVIT}
\EventYear{2016}
\EventDate{December 24--27, 2016}
\EventLocation{Little Whinging, United Kingdom}
\EventLogo{}
\SeriesVolume{42}
\ArticleNo{23}
\begin{document}

\maketitle              % typeset the header of the contribution
\begin{abstract}
Representation of Euclidean objects in a digital space has been a focus of research for over 30 years.  Digital line segments are particularly important as other digital objects depend on their definition (e.g., digital convex objects or digital star-shaped objects).  It may be desirable for the digital line segment systems to satisfy some nice properties that their Euclidean counterparts also satisfy.  The system is a consistent digital line segment system (CDS) if it satisfies five properties, most notably the subsegment property (the intersection of any two digital line segments should be connected) and the prolongation property (any digital line segment should be able to be extended into a digital line).  It is known that any CDS must have $\Omega(\log n)$ Hausdorff distance to their Euclidean counterparts, where $n$ is the number of grid points on a segment.  In fact this lower bound even applies to consistent digital rays (CDR) where for a fixed $p \in \mathbb{Z}^2$, we consider the digital segments from $p$ to $q$ for each $q \in \mathbb{Z}^2$.  In this paper, we consider families of weak consistent digital rays (WCDR) where we maintain four of the CDR properties but exclude the prolongation property.  In this paper, we give a WCDR construction that has optimal Hausdorff distance to the exact constant.  That is, we give a construction whose Hausdorff distance is 1.5 under the $L_\infty$ metric, and we show that for every $\epsilon > 0$, it is not possible to have a WCDR with Hausdorff distance at most $1.5 - \epsilon$.
\end{abstract}

\section{Introduction}
\label{sec:intro}
In this paper, we consider the digital representation of Euclidean objects.  For example, suppose we take a photograph of a convex object $O$.  Since $O$ is convex, we have that the Euclidean line segment connecting any two points inside of $O$ is contained within $O$.  If one wants an image segmentation algorithm to be able to find the pixels in our photograph that correspond to $O$, we may want the algorithm to be able to ensure that the output pixels are in some sense ``convex''.  How should we define a set of pixels to be convex?  The most natural way would be a similar definition to the Euclidean setting: a set of pixels is a digital convex object if the digital line segment connecting any two pixels is contained within the object.  This raises the question of how to define digital line segments.

In particular, consider unit grid $\mathbb{Z}^2$ where each point in the grid represents a pixel in an infinite image, and in particular consider the unit grid graph: for any two points $p=(p_x,p_y)$ and $q=(q_x,q_y)$ in $\mathbb{Z}^2$, $p$ and $q$ are neighbors if and only if $|p_x-q_x|+|p_y-q_y|=1$.  Now consider some pair of grid points $p$ and $q$; we would like to define a digital line segment $R_p(q)$ from $p$ to $q$ that is a path in the unit grid graph. There may be multiple ``good'' ways to digitally represent a Euclidean line segment.  For example, consider Figure \ref{fig:introFigs} (a).  In isolation, it may seem that either of Figure \ref{fig:introFigs} (b) or (c) would be a fine representation.  However when considering it within a family of digital line segments, the choice we make could impact the results for segmentation algorithms that rely on them.  In particular, the most simple definitions of digital line segments (``rounding'' the Euclidean line segment to the closest pixel) have a potentially undesirable property that their intersections may not be connected.  See Figure \ref{fig:introFigs2}.  In scenarios where we are not concerned with individual  digital line segments but rather multiple digital line segments (e.g., we are interested in digital convex objects), it may be desirable to consider carefully constructed digital line segment systems that satisfy some nice properties.

\begin{figure}
\centering
\begin{tabular}{c@{\hspace{0.1\linewidth}}c@{\hspace{0.1\linewidth}}c}

\includegraphics[scale=1]{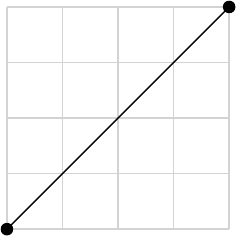}&
\includegraphics[scale=1]{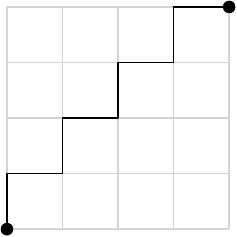}&
\includegraphics[scale=1]{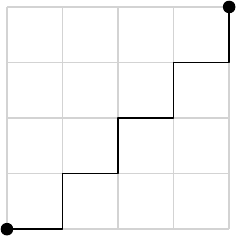}
\\
(a) & (b) & (c)
\end{tabular}
\caption{ (a) Euclidean line to digitize.  (b) One option. (c) Another option. }
\label{fig:introFigs}
\end{figure}

\begin{figure}
\centering
\includegraphics[scale=0.8]{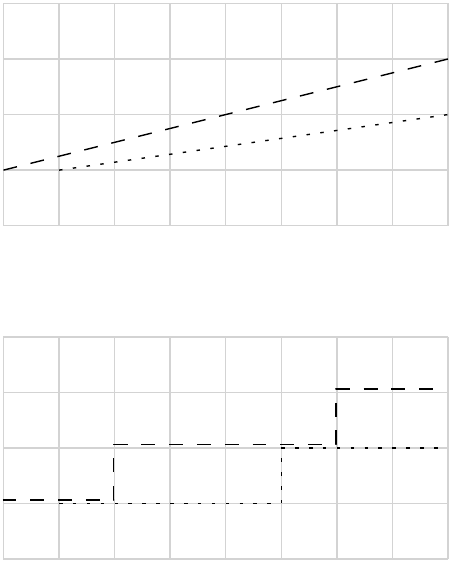}
\caption{Rounding two Euclidean line segments.}
\label{fig:introFigs2}
\end{figure}

For any point $o\in \mathbb{Z}^2$,  we call the set of all digital line segments $R_o(p)$ for each $p \in \mathbb{Z}^2$ a \textit{digital ray system} $R_o$.  Intuitively a digital ray system is a family of digital line segments that all share $o$ as a common endpoint.  A \textit{digital line segment system} has $R_p(q)$ defined for every $p,q\in\mathbb{Z}^2$.

\subsection{Consistent Digital Line Segments}
To deal with these issues, past researchers have considered systems of digital rays and digital line segments that collectively satisfy some properties that are also satisfied by their Euclidean counterparts.  In particular \cite{ChunKNT09,ChristPS12,ChowdhuryG15,ChowdhuryG16,ChiuK18} has considered systems that satisfy the following five properties.

\begin{description}
\item (S1) \textit{Grid path property:} For all $p, q \in \mathbb{Z}^2$, $R_p(q)$ is the points of a path from $p$ to $q$ in the grid topology.
\item (S2) \textit{Symmetry property:} For all $p, q \in \mathbb{Z}^2$, we have $R_p(q) = R_q(p)$.
\item (S3) \textit{Subsegment property:} For all $p, q \in \mathbb{Z}^2$ and every $r,s \in R_p(q)$, we have $R_r(s) \subseteq R_p(q)$.
\item (S4) \textit{Prolongation property:} For all $p, q \in \mathbb{Z}^2$, there exists $r \in \mathbb{Z}^2$, such that $r \notin R_p(q)$ and $R_p(q) \subseteq R_p(r)$.
\item (S5) \textit{Monotonicity property:} For all $p,q \in \mathbb{Z}^2$, if $p_x = q_x = c_1$ for any $c_1$ (resp. $p_y = q_y = c_2$ for any $c_2$), then every point $r \in R_p(q)$ has $r_x=c_1$ (resp. $r_y=c_2$).
\end{description}

Properties (S2) and (S3) are quite natural to ask for; the subsegment property (S3) is motivated by the fact that the intersection of any two Euclidean line segments is connected, and this property is violated by simple rounding schemes.  The prolongation property (S4) is motivated by the fact that any Euclidean line segment can be extended to an infinite line, and we may want a similar property to hold for our digital line segments.  While (S1)-(S4) form a natural set of axioms for digital segments, there are pathological examples of segments that satisfy these properties which we would like to rule out. For example, Christ et al.~\cite{ChristPS12} describe a CDS where a double spiral is centered at some point in $\mathbb{Z}^2$, traversing all points of $\mathbb{Z}^2$.  A CDS is obtained by defining $R_p(q)$ to be the subsegment of this spiral connecting $p$ and $q$.   To rule out these CDSes, property (S5) was added.

\begin{figure}
\centering
\begin{tabular}{c@{\hspace{0.1\linewidth}}c}

\includegraphics[scale=0.9]{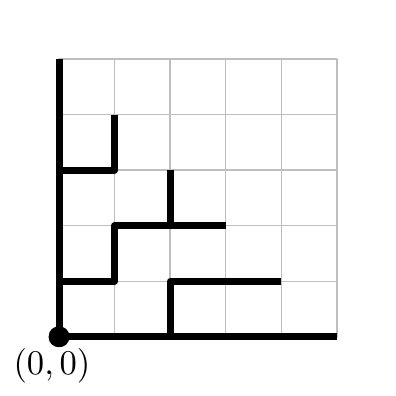}&
\includegraphics[scale=0.9]{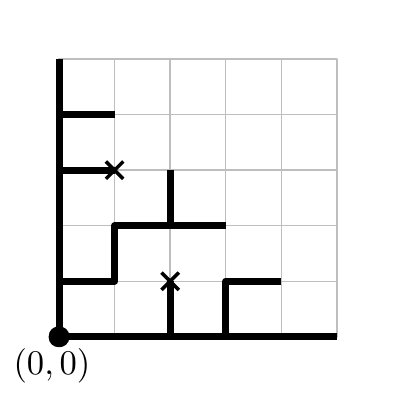}
\\
(a) & (b) 
\end{tabular}
\caption{ (a) A CDR for $(0,0)$ that satisfies (S4).  (b) A set of digital rays from $(0,0)$ that do not satisfy (S4).  In particular, the segments $R_{(0,0)}((1,3))$ and $R_{(0,0)}((2,1))$ do not prolong.}
\label{fig:prolongFigs}
\end{figure}

A digital ray system that satisfies properties (S1)-(S5) is called a \textit{consistent digital ray system} or CDR for short.  A digital line segment system that satisfies (S1)-(S5) is called a \textit{consistent digital line segment system} or CDS for short.  A CDR $R_o$ (or the rays for a single point in a CDS) can be viewed as a tree rooted at $o$, where the segment $R_o(p)$ is the unique simple path between $o$ and $p$ in the tree.  Note that the segments must be a tree because of property (S3).  See Figure \ref{fig:prolongFigs} (a) for an example CDR for $(0,0)$ to all grid points $(x,y)$ such that $x \geq 0, y \geq 0$, and $x+y \leq 5$.  See Figure \ref{fig:prolongFigs} (b) for an example of digital segments that satisfy all properties except for (S4).  Since they still satisfy (S3), the rooted tree perspective still applies to these segments.

\subsubsection{Previous Work on CDSes and CDRs}
 Luby \cite{Luby88} considers \textit{grid geometries} which are equivalent to systems of digital line segments satisfying (S1), (S2), (S5) described in this paper, and various works have considered CDRs and CDSes as defined above.  Past work has shown that there are many different CDR and CDS systems, so how should we measure the quality of such a system?  Past work usually measures the error of a system by considering the Hausdorff distance of a system.  For each grid point $v\in R_p(q)$, the distance from $v$ to the Euclidean line segment $\overline{pq}$ is usually defined to be the Euclidean distance between $v$ and the closest point to $v$ on $\overline{pq}$.  Then the error of $R_p(q)$, which we denote $E(R_p(q))$, is the maximum distance from $v$ to $\overline{pq}$ over all $v \in R_p(q)$. The error of a CDR or a CDS is then defined to be $\sup_{p,q \in \mathbb{Z}^2} \{ E(R_p(q)) \}$.  Chun et al.~\cite{ChunKNT09} give an  $\Omega(\log n)$ lower bound on the error of a CDR where $n$ is the number of grid points on a digital segment (i.e., $n := |p_x - q_x| + |p_y - q_y|$ for $R_p(q)$) which of course also applies to CDSes.  Chun et al.~give a construction of CDRs that satisfy the desired properties (S1)-(S5) with a tight upper bound of $O(\log n)$ on the error.  Note that the lower bound is due to combining properties (S3), (S4) and (S5).  For example, if we are willing to drop property (S3) then digital line segment systems with $O(1)$ error are easily obtained, for example the trivial ``rounding'' scheme used in Figure \ref{fig:introFigs} (d).  This system clearly satisfies (S5), and we can see that it will satisfy (S4) as well.  Without loss of generality assume that $p = (0,0)$.  Then $R_p(q)$ will clearly extend to $R_p(r)$ where $r = (2q_x, 2q_y)$ as $p$, $q$, and $r$ are co-linear.  Chun et al.~\cite{ChunKNT09} also show that if (S5) is relaxed, then $O(1)$ error is possible, although they describe the segments as ``locally snake-like almost everywhere'' but the segments may be acceptable if the resolution of the grid is sufficiently large.  Christ et al.~\cite{ChristPS12} extend the upper bound results from CDRs to CDSs by giving an optimal $O(\log n)$ upper bound the error for a CDS in $\mathbb{Z}^2$.  Chowdhury and Gibson have a pair of papers \cite{ChowdhuryG15,ChowdhuryG16} providing a characterization of CDSes in $\mathbb{Z}^2$. 

Most of the previous works listed above only apply to two-dimensional grids, but each of the properties (S1)-(S5) have natural generalizations to higher dimensions, and we may be interested in computing CDRs and CDSes for higher dimensions.  The construction of Chun et al.~\cite{ChunKNT09} for two-dimensional CDRs can be extended to obtain an $O(\log n)$ construction for a CDR in a three-dimensional grid.  More recently Chiu and Korman \cite{ChiuK18} have considered extending the two-dimensional results of \cite{ChristPS12} to three dimensions, and they show that at times they are able to obtain three-dimensional CDRs with error $\Omega(\log n)$, and even at times they can obtain a three-dimensional CDS, but unfortunately these systems have error $\Omega(n)$.

\subsection{Weak Consistent Digital Rays}
Suppose we have a CDR $R_o$ in $\mathbb{Z}^3$ where $o = (0,0,0)$, and suppose we consider the two-dimensional ``slice'' of points $v = (v_x, v_y, v_z)$ such that $v_z = 0$.  Now consider the digital rays $R_o(v)$ for each such $v$.  These rays must satisfy (S1), (S2), (S3), and (S5).  If any of these properties would be violated then the original CDR system would have to violate the same property; however, the two-dimensional slice does \textit{not} have to satisfy (S4).  Indeed, there may be a $v$ with $v_z = 0$ such that the digital segment $R_o(v)$ does not ``extend'' to any other point $v'$ such that $v'_z = 0$ but instead extends ``up'' to the point $(v_x, v_y, v_z + 1)$.  

This sparked the initial interest in what are called \textit{weak consistent digital rays} (WCDR) where the segments should satisfy all of the CDR properties except (S4).   In particular, the $\Omega(\log n)$ error lower bound of \cite{ChunKNT09} for two-dimensional CDRs critically relies upon (S4).  Consider a (not-weak) CDR $R_o$ in two-dimensions, and let $Q_1$ denote the ``first quadrant'' of $o$, that is $v\in Q_1$ if and only if $v_x \geq 0$ and $v_y \geq 0$.  For any $d\in \mathbb{N}$, let \textit{diagonal $d$} denote the Euclidean line $x + y = d$.  For any $v\in Q_1$, we say that $v$ is \textit{on} diagonal $v_x+v_y$. Consider the ``extension'' of $R_o$ from the points on diagonal $d$ to the points on diagonal $d+1$.  Chun et al.~show that there must be exactly 1 \textit{split point} $s$ on diagonal $d$ such that $R_o(s)$ extends to both $(s_x, s_y +1)$ and $(s_x+1,s_y)$.  For an example, see Figure \ref{fig:extend} (a) which shows the extensions of the CDR from Figure \ref{fig:prolongFigs} (a) from diagonal 4 to diagonal 5.  Then for every point $v$ on diagonal $d$ such that $v_x < s_x$, it only extends vertically to $(v_x, v_y+1)$, and every point $v$ on diagonal $d$ such that $v_x > s_x$ only extends horizontally to $(v_x+1, v_y)$.  This structure of CDRs helps lead to the error lower bound.  In the context of a WCDR, this split point property no longer holds.  Instead of picking a single split point on diagonal $d$ that then forces the extension of all other points on diagonal $d$, we can let each point on diagonal $d+1$ pick a ``parent'' point on diagonal $d$ (using the rooted tree perspective of a WCDR), and these parent selections do not need to be coordinated in any way since we do not have the requirement that segments on diagonal $d$ must extend to diagonal $d+1$.  See Figure \ref{fig:extend} (b) which shows the WCDR from Figure \ref{fig:prolongFigs} (b) between diagonals 4 and 5.  Points that are not chosen to be a parent are called \textit{inner leaves}.  Suppose on diagonal $d$ we have that the number of inner leaves is $x$.  Then it is not difficult to see that there will be $x+1$ split points on diagonal $d$, and when scanning the points on diagonal $d$ from left to right, we alternate between encountering split points and inner leaves, and the first and last will be split points.  This difference in structure creates the possibility of having $o(\log n)$ WCDRs, which in turn implies that it may be possible to obtain a (non-weak) CDR in $\mathbb{Z}^3$ with $o(\log n)$ error.  

\begin{figure}
\centering
\begin{tabular}{c@{\hspace{0.1\linewidth}}c}

\includegraphics[scale=0.8]{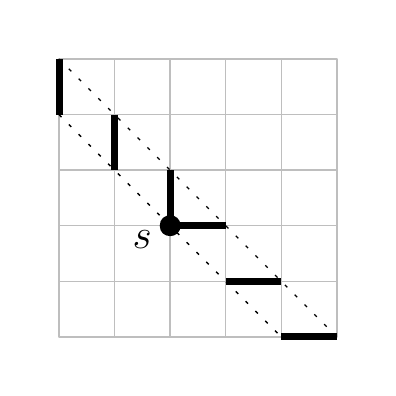}&
\includegraphics[scale=0.8]{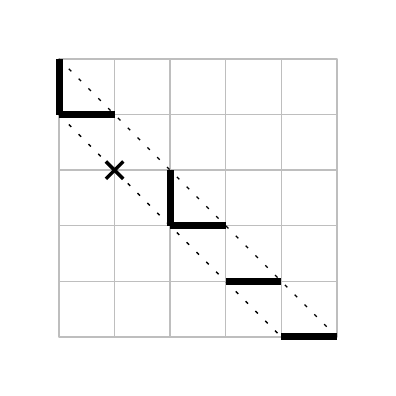}
\\
(a) & (b)
\end{tabular}
\caption{(a) A CDR extension with a single split point $s$.  (b) A WCDR extension with two split points and an inner leaf.}
\label{fig:extend}
\end{figure}

Chiu et al.~\cite{ChiuKST20} considered 2D WCDRs and considered their impact on (non-weak) CDRs in higher dimensions.  In particular, they consider the tradeoff between the number of inner leaves a 2D WCDR can have and the error of the system.  Indeed if a system does not have any inner leaves, then it is just a regular CDR and the $\Omega(\log n)$ error bound applies.  It may be possible to introduce some number of inner leaves to obtain a 2D WCDR with error $o(\log n)$.  How many inner leaves do we need to, say, obtain an error of $O(1)$?  They show that any WCDR defined for all points $p \in Q_1$ such that $p_x + p_y \leq N$ and has $k$ inner leaves between diagonals $N/2$ and $N$ has error $\Omega(\frac{N \log N}{N + k})$.  They then show the impact this has on (non-weak) CDRs in higher dimensions, as every inner leaf on a 2D slice must extend to a point not on this slice which will impact other slices.  They use this to show that any CDR in $d$ dimensions has error $\Omega(\log^{\frac{1}{d-1}} N)$.  They also consider what is the minimum number of inner leaves needed to obtain a WCDR with error $e$.  In the full version of their paper \cite{abs-2006-14059}, they give a system with error 2.5 under the $L_\infty$ metric establishing that $O(1)$ error is in fact possible.  Note that their $\Omega(\frac{N \log N}{N + k})$ lower bound implies that $k \in \Omega(N \log N)$ in order to achieve $O(1)$ error, but their construction has $k \in \Theta(N^2)$.  This leaves open the question the question as to whether is  possible to have a WCDR with $O(1)$ error and $o(N^2)$ inner leaves or if the lower bound could be improved.

\subsection{Our Contribution}
We consider optimizing the error of a 2D WCDR to the exact constant as we view the WCDR to be of general interest.  For some users, the $\Omega(\log n)$ lower bound that comes from including (S3), (S4) and (S5) may be unacceptable.  If we wish to achieve $o(\log n)$ error, then we are forced to drop at least one of these properties.  For a user who elects to drop (S3), there are plenty of options available that achieve $O(1)$ error and also satisfy (S4) and (S5) (e.g., a greedy rounding strategy).  It's possible to only drop (S5) and obtain $O(1)$ error, but a drawback of this system is the ``locally snake-like'' property that causes the segments to be of different ``widths'' on different diagonals (e.g, a segment may pass through 1 point on one diagonal but it passes through 3 points on another diagonal). The WCDR is the option for the user who does not want $\Omega(\log n)$ error but wants (S3) and (S5).

Since we want to optimize the exact constant, we need to pick the error metric carefully (i.e., use $L_\infty$, $L_2$, etc.).  Chiu et al.~used $L_\infty$ metric in their 2.5 error construction, and we believe that the $L_\infty$ metric indeed is the metric that best captures the error of a system.  That is, the diagonals 1 to $p_x + p_y$ form a kind of parametrization of $\overline{op}$, and when picking a point $v$ on diagonal $d$ to be on $R_o(p)$, we argue the goal should be to minimize the distance of $v$ to the intersection of diagonal $d$ and $\overline{op}$.  Let $i$ denote this intersection point.  Using $L_\infty$, the point on $\overline{op}$ that is closest to $v$ will always be $i$ regardless which $v$ on diagonal $d$ we are considering.  This does not hold for the $L_2$ metric where the closest point of $\overline{op}$ to $(v_x,v_y)$ and $(v_x+1,v_y-1)$ could be two different points.  For this reason, we believe that when being careful with constant factors, it is best to consider $L_\infty$.

\subsubsection{Our Results.} In this paper, we give a tight bound on the error of a WCDR in 2D to the exact constant.  We prove the following two theorems.

\begin{theorem}
There is a WCDR in $\mathbb{Z}^2$ with error 1.5 in the $L_\infty$ metric.
\label{thm:upperBound}
\end{theorem}

\begin{theorem}
For every $\epsilon > 0$, there is no WCDR in $\mathbb{Z}^2$ with error at most $1.5-\epsilon$ in the $L_\infty$ metric.
\label{thm:lowerBound}
\end{theorem}

We give a WCDR construction $R_o$ such that for every $p\in \mathbb{Z}^2$, we have that the error of $R_o(p)$ is less than 1.5 in the $L_\infty$ metric, and we show that for any $\epsilon>0$ that it is not possible to have a WCDR in 2D with error at most 1.5-$\epsilon$ in the $L_\infty$ metric.  Essentially, as the length of the segments gets larger, the error of our construction approaches 1.5 but never reaches it, and our lower bound shows that it is not possible to do better than this.

To state our results in the context of the work of Chiu et al.~\cite{ChiuKST20}, recall they showed that for all segments with length $N$ and $k$ inner leaves between diagonals $N/2$ and $N$ has error $\Omega(\frac{N \log N}{N + k})$.  This implies that $\Omega(N \log N)$ inner leaves are required to obtain $O(1)$ error.  We remark that for our construction, $k \in \Theta(N^2)$.  Our goal was to optimize the error and not minimize the number of inner leaves, but we are not aware of a ``simple'' way to modify our construction to obtain $o(N^2)$ inner leaves while maintaining $O(1)$ error.

\subsection{Organization of the Paper}
In Section \ref{sec:ub}, we present our construction with optimal error that proves Theorem \ref{thm:upperBound}.  In Section \ref{sec:lb}, we present the lower bound that proves Theorem \ref{thm:lowerBound}.

\section{Upper Bound}
\label{sec:ub}
In this section, we present a WCDR construction with error 1.5.  We begin with some preliminaries and definitions.

\subsection{Preliminaries}
\label{sec:notation}
We now define some notation that we will use throughout the paper.  Let $o$ be the origin $(0,0)$.   
For any point $p\in Q_1$ such that $p \neq o$, let $D(p) := p_x + p_y$ denote the diagonal that $p$ is on, and let $\ell(p)$ denote the Euclidean line through $o$ and $p$.  We define $M(p) := \frac{p_y}{p_x}$ to be the slope of $\ell(p)$ if $p_x > 0$, and otherwise we define $M(p)$ to be $\infty$.  Let $p^\leftarrow$, $p^\downarrow$, $p^\rightarrow$, and $p^\uparrow$ denote the points $(p_x - 1, p_y), (p_x, p_y-1), (p_x+1,p_y),$ and $(p_x,p_y+1)$ respectively.  Let $p^\nwarrow$ and $p^\searrow$ denote the points $(p_x - 1, p_y + 1)$ and $(p_x + 1, p_y - 1)$ respectively.  For two points $p$ and $q$, we say $p$ is \textit{above} $\ell(q)$ if $M(p) > M(q)$, and we say $p$ is \textit{below} $\ell(q)$ if $M(p) \leq M(q)$ (i.e., we break ties by saying $p$ is below $\ell(q)$.  We say that $p$ is \textit{between} $\ell(q)$ and $\ell(q')$ if $p$ is below $\ell(q)$ and is above $\ell(q')$.

\begin{figure}
\centering
\includegraphics[scale=0.8]{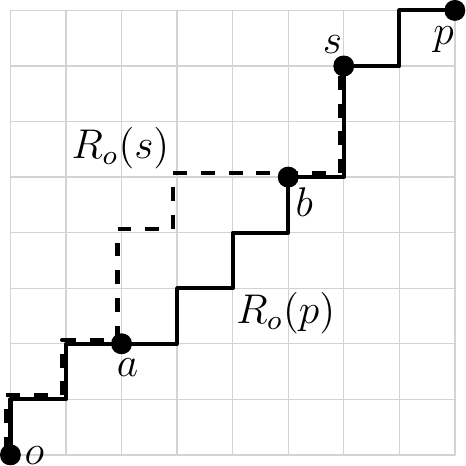}
\caption{An example of $R_o(s)$ and $R_o(p)$ splitting apart at $a$ and coming back together at $b$. The dashed line represents $R_o(s)$ and the solid line represents $R_o(p)$.}
\label{fig:s3}
\end{figure}

We can view any WCDR as a binary tree rooted at $o$, and then in this setting each point $v \in Q_1$ such that $v \neq o$ will have a ``parent'' on diagonal $D(v)-1$.  From this perspective,  the following procedure will produce a WCDR $R_o$ in $Q_1$.  For each point $v \in Q_1 \setminus \{o\}$, if $v_x = 0$ set $v.parent = v^\downarrow$, else if $v_y = 0$ set $v.parent = v^\leftarrow$, else arbitrarily choose one of $v^\downarrow$ and $v^\leftarrow$ to be $v.parent$.  Then the digital ray $R_o(p)$ can be computed in ``reverse'' by starting at $p$ and following the parents back to $o$.  This procedure clearly satisfies (S1), (S2), and (S5). We can show that (S3) will be satisfied with a simple proof by contradiction. Assume there exists an $s \in R_o(p)$ such that $R_o(s) \not\subseteq R_o(p)$. As pointed out in \cite{ChristPS12} there must be a point $a$ where $R_o(s)$ and $R_o(p)$ split apart for the first time, and a point $b$ when they first come back together, see Figure \ref{fig:s3}. Note in the figure $b.parent = b^\leftarrow$ in $R_o(s)$ and $b.parent = b^\downarrow$ in $R_o(p)$. However, our procedure only allows $b$ to select one parent for all segments that pass through $b$, a contradiction. Hence, (S3) holds.

Therefore a construction that was obtained from this procedure will certainly produce a feasible WCDR, and then the goal can be to carefully choose the parents so as to minimize the error.  After a WDCR for $Q_1$ is obtained, it can easily be extended to $\mathbb{Z}^2$ by ``mirroring'' the construction to the other quadrants.  We remark that regular CDRs that satisfy (S4) also can be viewed as a rooted binary tree in this way, but the key difference between the two problems is that if we want (S4), we have to ensure that every point on diagonal $d-1$ gets picked to be the parent for some point on diagonal $d$.  We have no such restriction when considering WCDRs; points on diagonal $d-1$ that have no points on diagonal $d$ that pick them to be a parent will be inner leaves.

\subsection{The construction}

In this section we describe a method for each point in $Q_1$ to pick its parent so that the resulting WCDR has error 1.5.  We maintain a pattern on each diagonal $d$ such that $d$ is a power of 2 starting with diagonal 4.   In particular, let $p$ be a point on diagonal $2^i$ for an integer $i \geq 2$.  Then $p$ will be a split point that will extend to points on diagonal $2^{i+1}$ if and only if $p_x$ is odd.  If $p_x$ is even such that $0 < p_x < 2^i$, then $p$ will be an inner leaf.  If $p_x = 0$ or $p_x = 2^i$, then it is on the $x$ or $y$ axis and only extends to the point of diagonal $2^{i+1}$ that is on the same axis.  See Figure \ref{fig:construct} where the split points are represented as squares and the inner leaves are represented as crosses.

\begin{figure}
\centering
\includegraphics[scale=1.0]{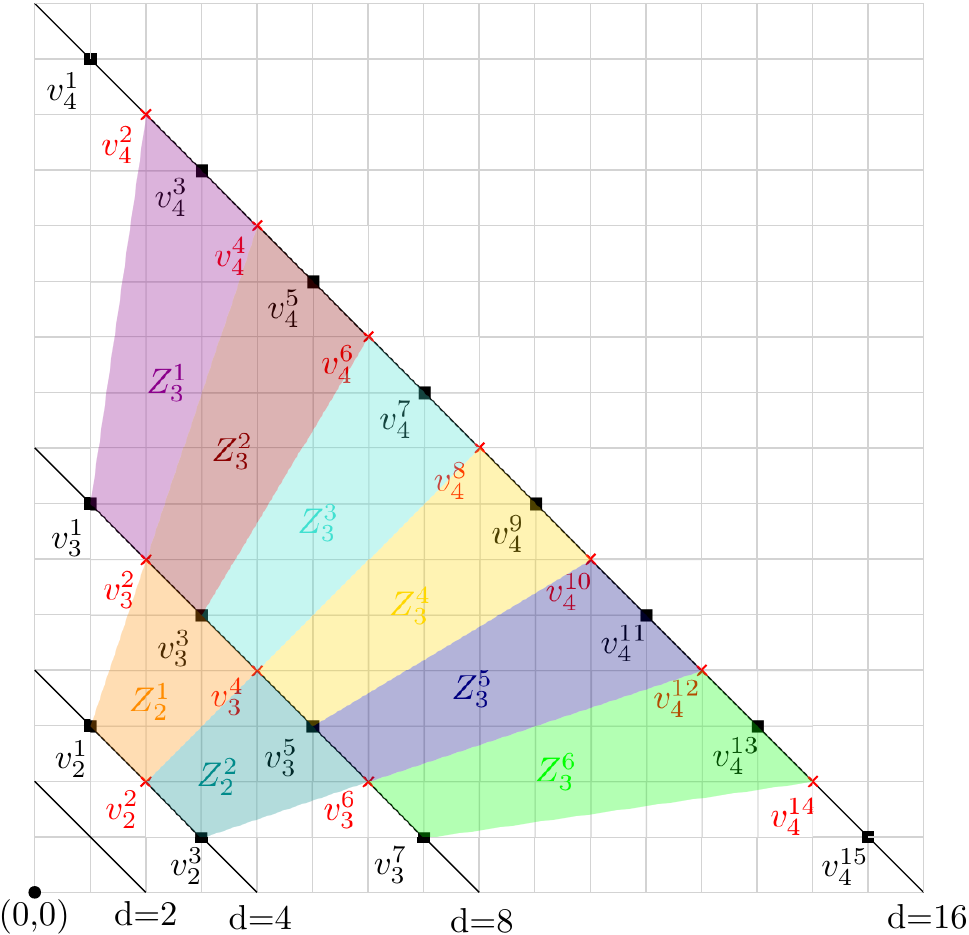}
\caption{Illustrating the definitions used in the construction.}
\label{fig:construct}
\end{figure}

\begin{figure}
\centering
\includegraphics[scale=1.0]{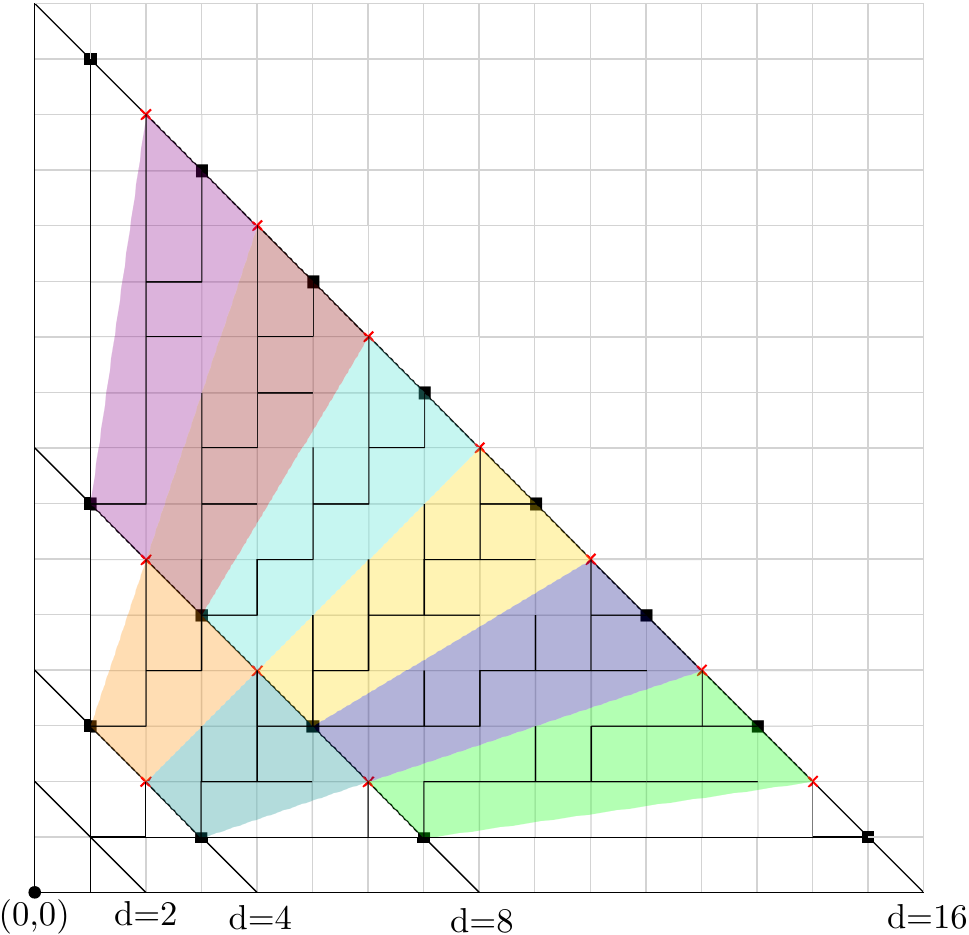}
\caption{Illustration of our construction of a WCDR with error 1.5 up to $d=16$.}
\label{fig:construct2}
\end{figure}

If $p$ is such that $p_x \leq 1$ or $p_y \leq 1$, then any monotone segment has error less than 1.  We will then focus on the points $p$ such that $p_x \geq 2$ and $p_y \geq 2$.  For each $i \geq 0$, let $C_i$ denote the points $p$ such that $2^i < D(p) \leq 2^{i+1}$, $p_x \geq 2$, and $p_y \geq 2$.  Let $v_i^j$ denote the point $(j, 2^i-j)$ (i.e., the point on diagonal $2^i$ with x-coordinate $j$) for each $j \in \{1, \ldots, 2^i-1\}$, and now consider the lines $\ell(v_i^1), \ell(v_i^2), \ldots \ell(v_i^{2^i-1})$.  Again suppose that $i \geq 2$ (so that there are at least 3 lines).  Then let $Z_i^j$ denote all $r\in \mathbb{R}^2$ such that $2^i < r_x + r_y \leq 2^{i+1}$, $r$ is below $\ell(v_i^j)$, and $r$ is above $\ell(v_i^{j+1})$.  We call each $Z_i^j$ a \textit{zone}.  Again see Figure \ref{fig:construct} where the shaded regions represent the zones. The intuition behind the zones is that they designate which points of $C_i$ we want the split points of diagonal $2^i$ to extend to.  Point $v_i^j$ will be a split point if and only if $j$ is odd.  Then we will have that all grid points in $Z_i^{j-1}$ and $Z_i^{j}$ will have $v_i^j$ as an ``ancestor''.  More specifically, as $v_i^j$ will be a split point, we will have $(v_i^j)^\uparrow$ and $(v_i^j)^\rightarrow$ will both pick $v_i^j$ as their parent.  Then all of the points of $Z_i^{j-1}$ will have $(v_i^j)^\uparrow$ as an ancestor and all the points of $Z_i^{j}$ will have $(v_i^j)^\rightarrow$ as an ancestor.  

Now let $M(i,j,d)$ for some diagonal $d \in \{2^i + 1, \ldots , 2^{i+1} - 1\}$ denote the ``midpoint'' of $Z_i^j$ with respect to diagonal $d$.  That is, it is the point (not necessarily with integer coordinates) on diagonal $d$ whose $L_\infty$ distance to $\ell(v_i^j)$ and $\ell(v_i^{j+1})$ is the same.  Also observe that for every diagonal $d \in \{2^i+1, 2^{i+1}\}$ there is either one or two grid points on diagonal $d$ that is in $Z_i^j$.  This follows from the fact that the distance between $\ell(v_i^j)$ and $\ell(v_i^{j+1})$ is 1 on diagonal $2^i$, and therefore for any such $d$ the distance between the lines is greater than 1 (implying there must be at least 1 grid point in $Z_i^j$ on diagonal $d$).  Also on diagonal $2^{i+1}$ the distance between the lines is 2, and therefore there cannot be 3 or more grid points on the diagonal (since we break ties in the same direction).

Now we are ready to formally state the construction.  For each $p \in Q_1$, we use Algorithm \ref{alg:construct} to pick its parent.  The digital ray $R_o(p)$ (which we will now call $R(p)$ for brevity) then is determined ``in reverse'' by going from $p$, $p.parent$, $(p.parent).parent$, etc. until we reach the origin. See Figure \ref{fig:construct2} for an illustration.  As mentioned before, this will certainly produce a feasible WCDR.  It remains to argue that the error of the resulting WCDR is 1.5.  We remind the reader that the error of the WCDR is the supremum of the errors of all digital rays.  In this section we prove that every digital ray in our WCDR has error less than 1.5.  Then in the next section we prove that it is not possible to have \textit{any} WCDR with error strictly better than 1.5 which implies that the supremum for our construction is in fact 1.5.

\begin{algorithm}
\caption{pickParent($p$)}
\label{alg:construct}
\begin{algorithmic}[1]
 \IF{$p_x \leq 1$ \AND $p$ is not $(1,0)$}
 \STATE $p.parent = p^\downarrow$
 \ELSIF{$p_y \leq 1$}
 \STATE $p.parent = p^\leftarrow$
 \ELSIF{$p_y$ is 2 \AND $D(p)$ is a power of 2}
 \STATE $p.parent = p^\downarrow$
 \ELSIF{$D(p)-1$ is a power of 2}
  \STATE Set $p.parent$ to be whichever of $p^\downarrow$ and $p^\leftarrow$ has odd x-coordinate.
 \ELSE
 \STATE Let $Z_i^j$ be the zone that $p$ belongs to.  Set $p.parent$ to be whichever of $p^\downarrow$ and $p^\leftarrow$ is closest to $M(i,j,D(p)-1)$, breaking ties arbitrarily.
 \ENDIF
\end{algorithmic}
\end{algorithm}

\begin{lemma}
The WCDR produced by Algorithm \ref{alg:construct} is such that for every $p\in Q_1$, $R(p)$ has error less than 1.5 in the $L_\infty$ metric.
\end{lemma}

\begin{proof}
Let $p$ be any point in $Q_1$.  If $p_x \leq 1$ or $p_y \leq 1$ then trivially any monotone digital ray will have error less than 1, so let us consider the points $p$ such that $p_x \geq 2$ and $p_y \geq 2$.  

In this paragraph, we handle the case where $p$ does not have a zone. The only such points $p$ that do not have a zone are the points $p$ such that $D(p)$ is a power of 2 (at least 4) and $p_y$ is 2 (this is due to the fact that we break ties by rounding down, and the region below these points are not in any zone).  Consider ``walking'' along $R(p)$ for these $p$ starting at $p$ and walking back towards $o$.  The segment moves vertically once, then moves horizontally until we reach $(1,1)$, then moves vertically once, then finally horizontally once to reach $o$.  Clearly this segment has error less than 1 for all such $p$.  

For the remainder of the proof, assume that $p$ does have a zone.  At a high level, we show that $R(p)$ will only contain grid points that are in zones that are intersected by $\ell(p)$ (until we reach a grid point with some coordinate that is 1).  This gets us most of the way there, but zones have diagonal widths that approach 2 as we reach the next power of 2 diagonal from the origin.  This means that it could be possible to have that $R(p)$ only contains grid points that are in zones intersected by $\ell(p)$ and yet the error approaches 2.  We will show that given our construction, the error of $R(p)$ will in fact be less than 1.5.

\subparagraph*{For every $q\in R(p)$ such that $q_x > 1$ and $q_y > 1$, $\ell(p)$ intersects the zone of $q$.}  Let $Z_i^j$ denote the zone that $p$ belongs to.  We first argue that when walking from $p$ to $o$ along $R(p)$, the segment ``stays inside'' $Z_i^j$ until we reach diagonal $2^i$.  This can be argued inductively: take any $q \in R(p)$ that is in $Z_i^j$.  If $D(q) = 2^i + 1$ then we are done, so suppose $D(q) > 2^i+1$.  Then $q$ picked its parent by choosing whichever of $q^\downarrow$ and $q^\leftarrow$ is closest to the midpoint of $Z_i^j$.  Since there must be at least one grid point on diagonal $D(q)-1$ in $Z_i^j$, it must be that at least one of $q^\downarrow$ and $q^\leftarrow$ is in $Z_i^j$.  Indeed, if $q^\downarrow$ is above $\ell(v_i^j)$ or $q^\leftarrow$ is below $\ell(v_i^{j+1})$ then $q$ would not be in $Z_i^j$, and if $q^\leftarrow$ is above $\ell(v_i^j)$ \textit{and} $q^\downarrow$ is below $\ell(v_i^{j+1})$ then there would not be any grid point on $D(q)-1$ in $Z_i^j$.  So $q$ has at least one parent option in $Z_i^j$, and a point in $Z_i^j$ clearly must be closer to the midpoint than a point outside $Z_i^j$, and therefore $q.parent$ will be in $Z_i^j$.  

Now consider the first point we encounter on $R(p)$ that is not in $Z_i^j$ when walking along $R(p)$ towards $o$.  This point is either $v_i^{j+1}$ or $v_i^j$ (depending on if $j$ is odd or even).  If this point has some coordinate that is 1, then we are done, so suppose it doesn't and therefore has a zone.  In the case where $j$ is even, then we reach $v_i^{j+1}$ ``from above''.  Note that the ``top line'' of $v_i^{j+1}$'s zone is the same as the ``top line'' of $p$'s zone (the points used to define the respective lines have the same slope).  Therefore $\ell(p)$ will intersect $v_i^{j+1}$'s zone.  Symmetrically, when $j$ is odd then we reach $v_i^j$ ``from the right'', and in this scenario we have that $\ell(v_i^{j+1})$ is the same line as the ``bottom line'' of $v_i^j$'s zone, and therefore $\ell(p)$ intersects $v_i^j$'s zone.  We then can apply these arguments inductively to see that $R(p)$ will stay inside the zones intersected by $\ell(p)$ until we reach a point with some coordinate that is 1. 

\subparagraph*{If $R(p)$ contains a point $q$ that is distance at least 1.5 to one of the lines $\ell$ defining its zone, then every $q'$ such that $D(q) < D(q') \leq D(p)$ is in the same zone and is distance at least 1.5 from $\ell$.}  Consider such a $q$, and let $Z_i^j$ denote the zone of $q$.  Note that by the above argument, it must be that $\ell(p)$ intersects $Z_i^j$.  If $q$ is $p$ then we are done, so suppose $q$ is not $p$.  Then some point $q' \in R(p)$ picked $q$ to be its parent.  We will first argue that $q'$ must also be in $Z_i^j$.  Of course it must be in a zone intersected by $\ell(p)$, and therefore it cannot be in $Z_i^{j-1}$ or $Z_i^{j+1}$.  Moreover the only point in $Z_i^j$ that is chosen to be the parent of a point in a different zone is the split point on diagonal $2^{i+1}$, but this point is distance 1 to both lines defining the zone, and therefore cannot be $q$.  This implies that $q'$ must also be in $Z_i^j$.  

Now suppose without loss of generality we have that the distance from $q$ to $\ell(v_i^{j+1})$ is at least 1.5.  We will show that the distance from $q'$ to $\ell(v_i^{j+1})$ is also at least 1.5.  We will do this by showing that $q'$ must be $q^\uparrow$.  Note that when $q'$ picked $q$ as its parent, it must have done so in the last line of Algorithm \ref{alg:construct} since $D(q)$ cannot be a power of 2.  Since the distance from $q$ to $\ell(v_i^{j+1})$ is at least 1.5 it must be that $q^\searrow$ is also in $Z_i^j$.  Moreover since the distance between $\ell(v_i^j)$ and $\ell(v_i^{j+1})$ is less than 2 on $D(q)$ it must be that the distance from $q$ to $\ell(v_i^j)$ is less than 0.5, which implies that the distance from $q^\searrow$ to $\ell(v_i^j)$ is less than 1.5.  So which point could have picked $q$ as its parent?  It could not have been $q^\rightarrow$, as it would have preferred $q^\searrow$ as its parent over $q$.  Therefore it must have been $q^\uparrow$ that picked $q$ as its parent, and $q^\uparrow$ will be farther from $\ell(v_i^{j+1})$ than $q$.

\subparagraph*{Putting it all together.}  Now consider $R(p)$, let $Z_i^j$ denote $p$'s zone, and consider any point $q \in R(p)$.  We will argue that the distance from $q$ to $\ell(p)$ is less than 1.5 on $D(q)$.  If $q$ is distance at least 1.5 from one of the lines for its zone, then by the previous argument we have that $p$ is in the same zone as $q$ and is distance more than 1.5 from the same line.  Since the ``width'' of the zone at $D(q)$ is less than 2, this implies that the distance from $q$ to $\ell(p)$ on $D(q)$ is less than 0.5.  So now assume that $q$ is less than 1.5 to both lines of its zone.  Since we know that $\ell(p)$ intersects the zone of $q$, it directly follows that the distance from $q$ to $\ell(p)$ on $D(q)$ is less than 1.5.
\end{proof}

\section{Lower bound}
\label{sec:lb}

\begin{figure}
\centering
\includegraphics[scale=0.7]{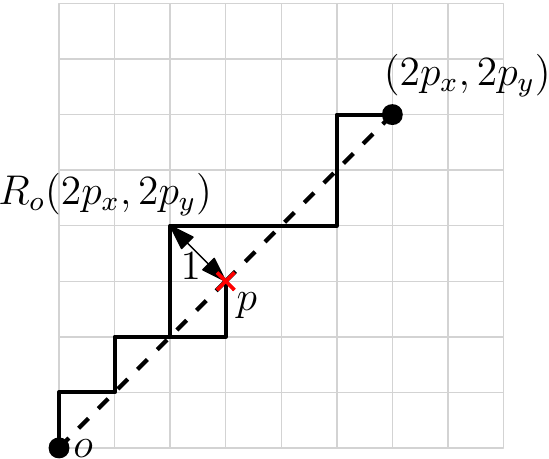}
\caption{Illustrating for any inner leaf $p$ the segment $R_o(2p_x,2p_y)$ must have error of at least $1$.}
\label{fig:error1}
\end{figure}

In this section, we prove Theorem \ref{thm:lowerBound} which implies that Algorithm \ref{alg:construct} produces a WCDR with an optimal error of 1.5. We begin with some trivial lower bounds, followed with a high level sketch of our proof of a tight lower bound, followed by the technical details.

\subsection{Trivial Lower Bounds}
There is no known previous work in obtaining a lower bound. There is a trivial lower bound of 0.5 for the segment $R((1,1))$.  There are two options for the segment (pass through either $(0,1)$ or $(1,0)$), and both of them have an error of 0.5.  We can also obtain a fairly easy lower bound of 1 by considering the effects of inner leaves.  If a WCDR does not have any inner leaves, then it satisfies (S4) as well and therefore is actually a CDR and therefore has $\Omega(\log n)$ error, so consider an WCDR, $R_o$, that has an inner leaf $p$.  We can show that it must have an error of at least 1 in the following way. Consider the segment $R_o((2p_x, 2p_y))$. Trivially, $\ell(2p_x, 2p_y)$ passes through $p$, but by assumption $p$ is an inner leaf and therefore $R_o((2p_x, 2p_y))$ must pass through a point on $D(p)$ that is either ``above'' $p$ or ``below'' $p$, see Figure \ref{fig:error1}. Hence, any WCDR must have error of at least 1.

So obtaining a lower bound of 0.5 is trivial, improving it to 1 is fairly simple, but improving it to 1.5 (which is tight given our construction) is more technical.

\subsection{Tight Lower Bound Preliminaries}

We begin with some definitions that we use in the lower bound proof.  Fix any WCDR in $\mathbb{Z}^2$.  Recall we are proving that for any $\epsilon > 0$, it is not possible for the WCDR to have error at most $1.5-\epsilon$.  For a sufficiently large integer $N$ (which depends on $\epsilon$), we will ``cut off'' the WCDR at diagonal $N$, obtaining a finite WCDR that is sufficient for proving the lower bound.  This allows us to use maximums and minimums in our definitions rather than supremums and infimums.  Similar to the previous section, we only consider the first quadrant $Q_1$ of $o$.  We say that a point $v$ \textit{extends} to diagonal $d > D(v)$ if there is some point $p$ with $D(p) = d$ such that $v \in R(p)$.  Similarly, if we say that $v$ extends to $p$ if $v \in R(p)$.  We let $Subtree(v)$ denote the set of all grid points that $v$ extends to.  For any subset $S$ of the grid points, we define $Cone(S)$ to denote all points that are between $\ell(t)$ and $\ell(r)$, where $t$ is the point in $S$ with maximum slope, and $r$ is the point in $S$ with minimum slope.  We define $ConeWidth(S,d)$ to be the distance between $\ell(t)$ and $\ell(r)$ on diagonal $d$.  For a split point $s$, we call $Subtree(s^\rightarrow)$ the \textit{bottom branch} of $s$, and we call $Subtree(s^\uparrow)$ the \textit{top branch} of $s$.

\begin{observation}

Given a set of points, $S$, the number of grid points between $Cone(S)$ on diagonal $d$ must be either $\lfloor ConeWidth(S,d) \rfloor$ or $\lfloor ConeWidth(S,d) \rfloor + 1$.

\label{obs:num_grid_points_between_cone}

\end{observation}

Let $Inter(\ell(p),d)$ be the point $\ell(p)$ intersects diagonal $d$. 

\begin{observation}
Given a point $p$ and a diagonal $d$, the x-coordinate of the 
$Inter(\ell(p),d)$ is $d\frac{p_x}{p_x + p_y}$, and the y-coordinate of
$Inter(\ell(p),d)$ is $d\frac{p_y}{p_x + p_y}$.

\label{obs:get_coordinates_slope_of_line}
\end{observation}

We will use the following lemma in our lower bound proof.

\begin{lemma}
Suppose we have a WCDR in $\mathbb{Z}^2$ with error less than 1.5, and further suppose there is a point $v$ that does not extend to diagonal $d$ for some $d > D(v)$.  Then $v^\nwarrow$ extends to diagonal $d$ (if it is in $Q_1$), and $v^\searrow$ extends to diagonal $d$ (if it is in $Q_1$).
\label{lem:noConsecutiveDeadNodes}
\end{lemma}

\begin{proof}
Suppose the contrary.  Without loss of generality, suppose that $v^\searrow$ also does not extend to diagonal $d$.  Then on some diagonal $d' > d$ there will be a grid point $p$ such that $M(p) = \frac{M(v) + M(v^\searrow)}{2}$ (because the slopes are rational).  Then $v$ and $v^\searrow$ are both distance 0.5 to $\ell(p)$, but $R(p)$ cannot contain $v$ or $v^\searrow$.  Therefore no matter which point from diagonal $D(v)$ is on $R(p)$, we must have that the error at that point is at least 1.5, a contradiction.
\end{proof}

\subsection{Proof Sketch}

We now give a high level overview of our lower bound proof.  Suppose we have any WCDR construction with error at most 1.5.  We will show that there is some point $p \in \mathbb{Z}^2$ such that 

\begin{enumerate}
    \item $p_y > p_x$ but $M(p)$ is ``very close'' to 1,
    \item $ConeWidth(Subtree(p), D(p))$ is ``very close'' to 2,
    \item and $p$ is ``very close'' to the center of $Cone(Subtree(p))$ on diagonal $D(p)$.
\end{enumerate}

Suppose we can show the existence of such a point $p$.  We will show that this implies that no matter which point $p$ picks as its parent, the error of that choice will be ``very close'' to 1.5.  See Figure \ref{fig:pRuinsEverythingAndconewidth} (a).

Let $t$ be the point in $Subtree(p)$ with maximum slope, and let $r$ be the point in $Subtree(p)$ with minimum slope.  From the assumptions on $p$, we have that $\ell(t)$ is a distance of close to 1 ``above'' $p$ on $D(p)$, and $\ell(r)$ is a distance of close to 1 ``below'' $p$ on $D(p)$. Moreover $M(p)$ is ``very close'' to 1 (which implies that $M(t)$ and $M(r)$ also are ``very close'' to 1 if $D(p)$ is ``sufficiently large'').  Whichever point $p$ picked as its parent will be on $R(t)$ and $R(r)$.  But if $p$ picks $p^\leftarrow$ as its parent, $p^\leftarrow$ will have a distance that is ``very close'' to 1.5 from $\ell(r)$.  If $p$ picks $p^\downarrow$ as its parent, $p^\downarrow$ will have a distance that is ``very close'' to 1.5 from $\ell(t)$.  Therefore the WCDR will have to have an error that is ``very close'' to 1.5 no matter which choice $p$ made.  

\begin{figure}
\centering
\begin{tabular}{c@{\hspace{0.1\linewidth}}c}

\includegraphics[scale=1.2]{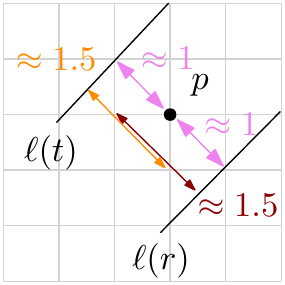}&
\includegraphics[scale=0.8]{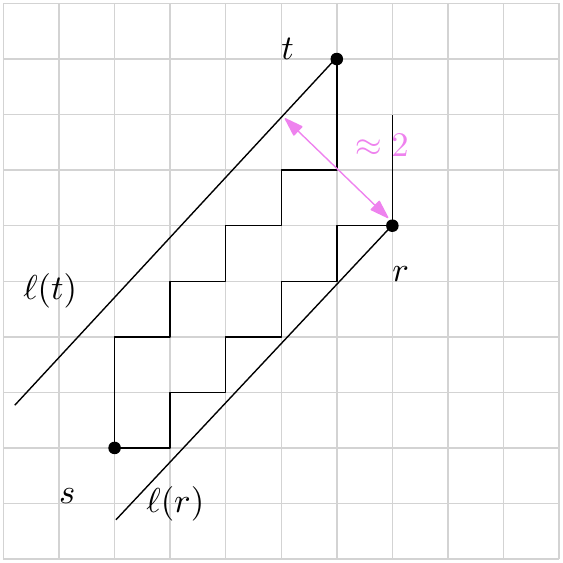}
\\
(a) & (b)
\end{tabular}
\caption{ (a) $p$ satisfies the three properties above and no matter which parent it picks, the error will be close to 1.5.  (b) $ConeWidth(Subtree(s), D(r))$ is close to 2.  $D(r) - D(s)$ will be sufficiently small so that $ConeWidth(Subtree(s), D(s))$ is also close to 2.}
\label{fig:pRuinsEverythingAndconewidth}
\end{figure}

Now to show that for any $\epsilon > 0$ that there is no WCDR with error at most $1.5 - \epsilon$, we pick a $\delta$ such that $0 < \delta < \min\{\epsilon,0.1\}$ and we show that such a point $p$ exists where ``very close'' is a function of $\delta$.  Then the above analysis will show that no matter which point $p$ picks as its parent, the error will have to be at least $1.5 - \delta$ > $1.5 - \epsilon$.  

We show that $p$ exists by showing that there must be a split point $s$ that satisfies:

\begin{itemize}
    \item $1 + \frac{2\delta}{3-2\delta} < M(s) < 1 + \frac{4\delta}{3-2\delta}$
    \item $ConeWidth(Subtree(s),D(s))$ is ``very close'' to 2.
\end{itemize}

Since $M(s) > 1$, eventually it will either have to make two consecutive vertical extensions in its bottom branch, or it will have to have a split point in $Subtree(s)$.  We can show that this will imply that $ConeWidth(Subtree(s),D(s))$ is ``very close'' to 2.  For example, see Figure \ref{fig:pRuinsEverythingAndconewidth} (b)
which shows a split point with two consecutive vertical extensions in its bottom branch.  It then follows that $s$ must extend to the points $t$ and $r$, where $r$ is the point at the ``bottom'' of the double vertical extension in the bottom branch of $s$, and $t$ is $((r^\uparrow)^\uparrow)^\nwarrow$.  Since $M(s), M(r)$ and $M(t)$ are all ``very close'' to 1, it follows that $ConeWidth(Subtree(s), D(r))$ is ``very close'' to 2, and we in fact will show that $ConeWidth(s, D(s))$ remains ``very close'' to 2 as $D(r)$ is sufficiently close to $D(s)$.  Recall we are looking for a point $p$ with 3 properties.  Certainly $s$ satisfies both properties (1) and (2) above, but it may not satisfy (3).  But we will show that we do not have to go ``too many'' diagonals before $D(s)$ before we must find a point $p \in R(s)$ that is close to the center of its cone and still satisfies (1) and (2).  In order to show this last part, we need $D(s)$ to be sufficiently large.

\subsection{Formal Proof}
For any $\epsilon > 0$, fix any $\delta$ such that $0 < \delta < \min\{\epsilon, 0.1\}$.  The following lemma proves the existence of the split point $s$ as we described in the proof sketch.  

\begin{lemma}
For any WCDR in $\mathbb{Z}^2$ with error less than 1.5 and any $0 < \delta \leq 0.1$ there is a split point $s$ that satisfies the following:

\begin{enumerate}
\item   For every $p \in Subtree(s)$, $1 + \frac{2\delta}{3-2\delta} < M(p) < 1 + \frac{4\delta}{3-2\delta}$
   
    \item $ConeWidth(Subtree(s),D(s)+ \frac{7+\delta}{\delta}) > 2-\frac{2\delta}{3}$
    \item $\frac{63}{\delta^2} \leq D(s) \leq \frac{382}{\delta^2 - \frac{2}{3}\delta^3} $ 
    
\end{enumerate}
\label{lem:lowerBound}
\end{lemma}

\begin{proof}

Fix any $0 < \delta \leq 0.1$. Suppose there exists a WCDR with error at most 1.5-$\delta$.  We prove the lemma by first showing how to pick the split point $s$ from the WCDR, and then we show that this choice of $s$ must satisfy the three conditions.

\subparagraph*{Picking the split point $s$.}  Let $x := \lceil \frac{200}{\delta^2} \rceil$, $y := x + \lceil \frac{600 - 182 \delta}{\delta (3 - 2\delta)} \rceil$, and $N := x+y$. Let $\Sigma$ denote the subset of consecutive grid points along diagonal $N$  from the grid point $(x-21,y+21)$ to $(x,y)$. Let $\ell(t)$ and $\ell(r)$ be the Euclidean lines as previously defined for $Cone(\Sigma)$. That is $t$ is the point in $\Sigma$ with maximum slope, and $r$ is the point in $\Sigma$ with minimum slope. See Figure \ref{fig:sigma_points}. We will first show that (3) from Lemma \ref{lem:lowerBound} must be satisfied. Let $N' := \lceil 20 \frac{N}{21} \rceil$, and let $\Sigma'$ denote the set of points between $\ell(t)$ and $\ell(r)$ on $N'$.
Notice that the width of $Cone(\Sigma)$ increases by $1$ every $\frac{N}{21}$ diagonals. That is to say that $ConeWidth(\Sigma,\frac{20}{21} N) = ConeWidth(\Sigma,N) -1 $ = 20.   
It follows  that $ConeWidth(\Sigma, N') \geq 20$ which implies from Observation \ref{obs:num_grid_points_between_cone} that $\vert \Sigma' \vert \geq 20$.

\begin{figure}
    \centering
    \includegraphics[scale=0.8]{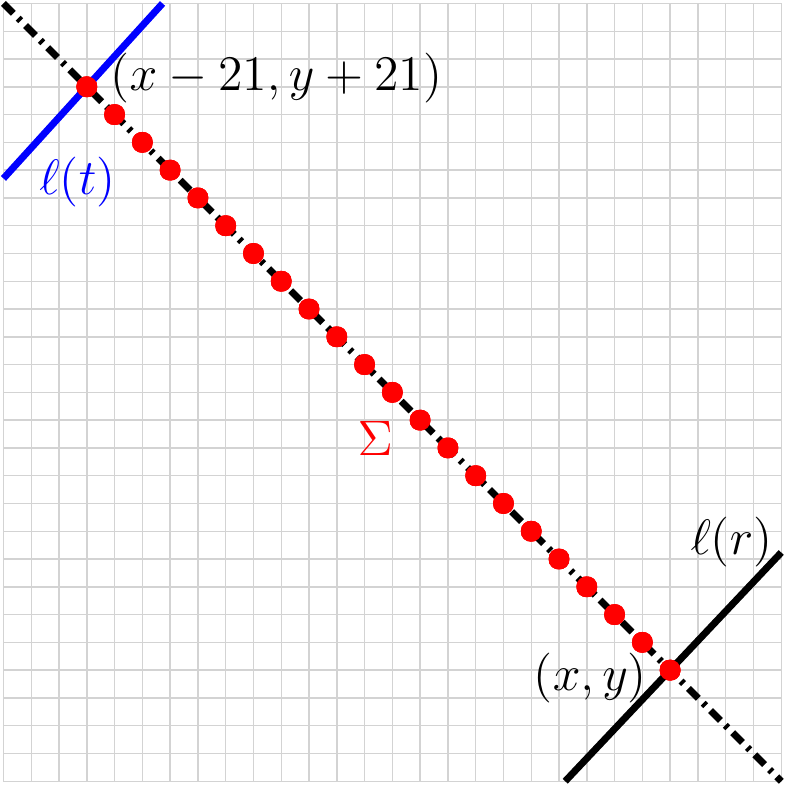}
    \caption{An illustration of the consecutive grid points in $\Sigma$.}
    \label{fig:sigma_points}
\end{figure}

We will lower bound the number of points in $\Sigma'$ that must extend to a point in $\Sigma$. Let $p$ be a point in $\Sigma'$.  There are three cases for why $p$ does not extend to a point in $\Sigma$. The first is that $p$ extends to a point above $\ell(t)$, the second is $p$ extends to a point below $\ell(r)$, and the third is that $p$ does not extend to $N$. 
Let us consider the first case. 
Let $p$ extend to a point $q$, such that $q \notin\Sigma$, $D(q) = N$, and $q$ is above $\ell(t)$.
Consider the point $w := (x-22, y+22)$. Notice that $w^\searrow$ is on $\ell(t)$ and $D(w) = N$. It follows that $M(q) \geq M(w)$.

 We would like to know the difference in x-coordinates  when $\ell(w)$ and  $\ell(w^\searrow)$ intersect $N'$.  
 Recall that $N' = \lceil \frac{20}{21}N \rceil$, which implies that $N' \geq  \frac{20}{21}N $. We know that the difference in x-coordinates of  $w^\searrow$ and $w$ is $1$.  It then follows that the difference in x-coordinates of $Inter(\ell(w^\searrow), N')$ and $Inter(\ell(w), N')$ is at least $\frac{20}{21}$. 
 Therefore error of $R(q)$ is at least the difference between $p_x$ and the x-coordinate of $Inter(\ell(w^\searrow), N')$ plus $\frac{20}{21}$. It follows that $p$ can only be the point with the largest slope in $\Sigma'$, otherwise error for $R(q)$ is greater than $1 + \frac{20}{21}$. See Figure \ref{fig:w_and_w_searrow}. We now analyze case two with a similar argument. Let $l:= (x+ 1, y - 1)$. If $p$ extends to a point on $N$ that is below $\ell(r)$, then error of $R(q)$ is at least the difference between the x-coordinate of $Inter(\ell(l^\nwarrow),N')$ and $p_x$ plus $\frac{20}{21}$. It then follows that $p$ can only be the point with the smallest slope in $\Sigma'$, otherwise error for $R(q)$ is greater than $1 + \frac{20}{21}$. Therefore there are only two points in $\Sigma'$ that can extend to a point on $N$ that is not in $\Sigma$. It follows that there are at least $20 - 2 = 18$ points in $\Sigma'$ that either extend to a point in $\Sigma$ or do not extend to $N$. 
 If a point $p \in \Sigma'$ does not extend to $N$ then from Lemma \ref{lem:noConsecutiveDeadNodes} it must be that $p^\nwarrow$ and $p^\searrow$ extend to $N$. It then follows that the number of points that must extend to $\Sigma$ is at least $\frac{18}{2}= 9$.

\begin{figure}
    \centering
    \includegraphics[scale=1]{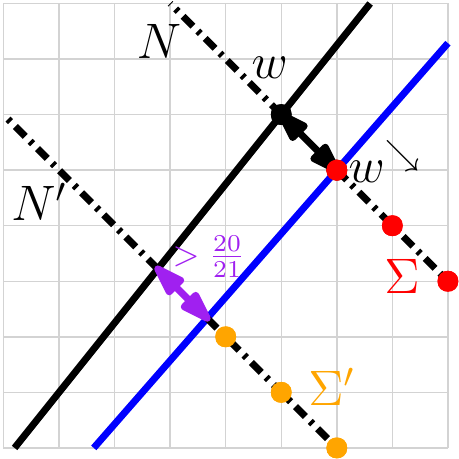}
    \caption{An example of why $p$ must be the point with the highest slope in $\Sigma'$. Note that the difference between $N$ and $N'$ is not drawn to scale.}
    \label{fig:w_and_w_searrow}
\end{figure}

Next we will upper bound the number of points on diagonal $N'' := \lfloor 5\frac{N}{21} \rfloor$ that can extend to a point in $\Sigma$. Let $p$ be a point on $N''$ that extends to $\Sigma$. Let $\alpha$ be a cone such that the top edge is $1.5 - \delta$ above $\ell(t)$ on $N''$ and the bottom edge is $1.5 - \delta$ below $\ell(r)$ on $N''$.  If $p$ is more than $1.5 - \delta$ above $\ell(t)$ or more than $1.5 - \delta$ below $\ell(r)$ then the error is more than $1.5 - \delta$. It follows that $p$ must be between the edges of $\alpha$. Recall that the width of $Cone(\Sigma)$ increases by $1$ every $\frac{N}{21}$ diagonals. It follows that $ConeWidth(\Sigma, N'') \leq 5$ which implies that the cone width of $\alpha$ on $N''$ is at most $5 + 3 - 2\delta$. Therefore the cone width of $\alpha$ on $N''$ is less than $8$. It follows from Observation \ref{obs:num_grid_points_between_cone} that the number of grid points between $\alpha$ on $N''$ is at most $8$.  Therefore there are at most $8$ points on $N''$ that can extend to $\Sigma$. See Figure \ref{fig:alpha_cone}.

\begin{figure}
    \centering
    \includegraphics[scale=0.8]{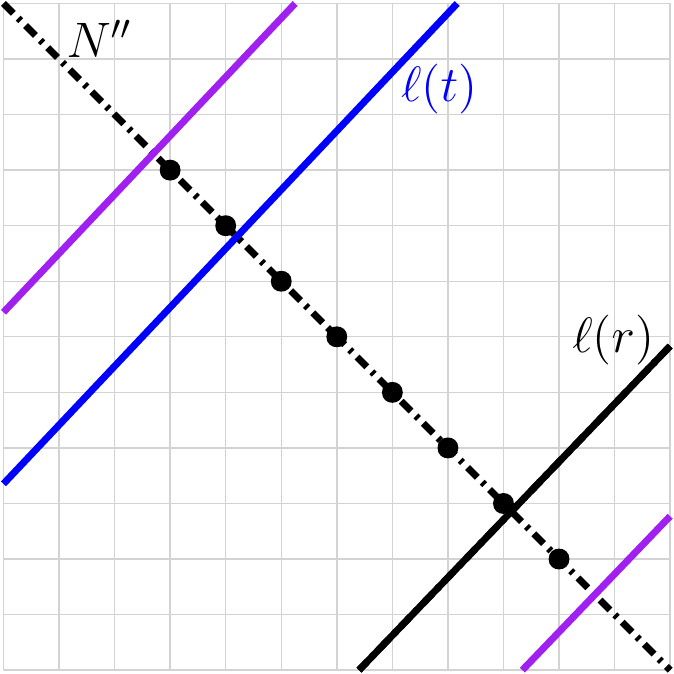}
    \caption{An example of the maximum points on $N''$ that can extend to $\Sigma$. $\alpha$ is represented by the outside purple lines.}
    \label{fig:alpha_cone}
\end{figure}

Since there is at least $9$ points in $\Sigma'$ that extend to a point in $\Sigma$, and there is at most $8$ points on $N''$ that extend to a point in $\Sigma$, it follows that there must be a split point $s$ such that $N'' \leq D(s) < N'$, that has both $s^\uparrow$ and $s^\rightarrow$ extend to at least one point in $\Sigma$.  This is the split point $s$ that we will now show satisfies the three conditions of the lemma.

\subparagraph*{Proving $s$ satisfies (3).}
We have $N'' \leq D(s) < N'$.  We first will show that $\frac{63}{\delta^2} \leq N''$ by showing that $\frac{63}{\delta^2} \leq \lfloor \frac{5}{21}N \rfloor$ is true. We have $N = 2\lceil \frac{200}{\delta^2} \rceil + \lceil \frac{600 - 182\delta}{\delta(3-2\delta)} \rceil \geq \frac{272}{\delta^2}$ for all $\delta \in (0, 0.1]$, and therefore $\frac{63}{\delta^2} \leq \lfloor \frac{5}{21}(\frac{272}{\delta^2}) \rfloor$. Thus $\frac{63}{\delta^2} \leq \lfloor \frac{5}{21}N \rfloor$ is true which implies $ \frac{63}{\delta^2} \leq N''$.

Next we want to show that $N' \leq  \frac{382}{\delta^2 - \frac{2}{3}\delta^3} $ by showing that $\lceil \frac{20}{21}N \rceil \leq \frac{382}{\delta^2 - \frac{2}{3}\delta^3}$. We have $N = 2\lceil \frac{200}{\delta^2} \rceil + \lceil \frac{600 - 182\delta}{\delta(3-2\delta)} \rceil \leq \frac{400}{\delta^2 - \frac{2}{3}\delta^3}$ for all $\delta \in (0, 0.1]$, and $\lceil \frac{20}{21} (\frac{400}{\delta^2 - \frac{2}{3}\delta^3}) \rceil \leq \frac{382}{\delta^2 - \frac{2}{3}\delta^3}$. Thus $\lceil \frac{20}{21}N \rceil\leq \frac{382}{\delta^2 - \frac{2}{3}\delta^3}$ which implies $N' \leq \frac{382}{\delta^2 - \frac{2}{3}\delta^3}$.
This concludes the proof of (3) from Lemma \ref{lem:lowerBound}.

\subparagraph*{Proving $s$ satisfies (1).}
While it is true that both $s^\uparrow$ and $s^\rightarrow$ must extend to at least one grid point in $\Sigma$, it is possible that they also extend to grid points outside of $\Sigma$ on $N$. However, any $\ell(p)$, such that $p \in Subtree(s)$, must cross $D(s)$ at most $1.5 - \delta$ above or below $s$. Moreover, since $s$ extends to at least two grid points in $\Sigma$ it must be that $s$ can be at most $1.5 - \delta$ above $\ell(t)$ or at most $1.5 - \delta$ below $\ell(r)$. It follows that any $\ell(p)$, such that $p \in Subtree(s)$, must cross $D(s)$ at most $3 - 2\delta$ above $\ell(t)$ or at most $3 - 2\delta$ below $\ell(r)$. Recall that we just showed that (3) from Lemma \ref{lem:lowerBound} is true. Therefore to prove (1) from Lemma \ref{lem:lowerBound} we will show that $M(p) < 1 + \frac{4\delta}{3-2\delta}$ for every point $p$ that is $3$ above $\ell(t)$ such that $N'' \leq D(p) < N'$,  and that $M(p) > 1 + \frac{2\delta}{3 - 2\delta}$ for every point $p$ that is $3$ below $\ell(r)$ such that $N'' \leq D(p) < N'$.

\begin{figure}
    \centering
    \includegraphics[scale=1]{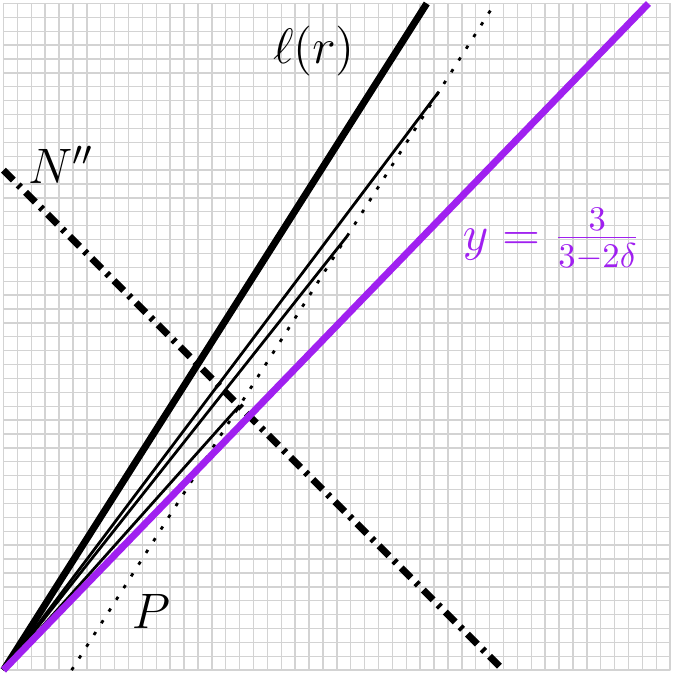}
    \caption{An illustration of $P$ under $\ell(r)$. The Euclidean line with slope $\frac{3}{3-2\delta}$ is represented by the purple line on the bottom.}
    \label{fig:slope_gets_smaller}
\end{figure}

We will begin by first showing that $M(p) > 1 + \frac{2\delta}{3 - 2\delta} = \frac{3}{3-2\delta}$ for every point $p$ that is $3$ below $\ell(r)$ such that $N'' \leq D(p) < N'$. Let $p \in P$ be the set of all points such that $p$ is $3$ below $\ell(r)$ and $N'' \leq D(p) < N'$. Notice that the point $p \in P$ with minimum slope is such that $D(p) = N''$. See Figure \ref{fig:slope_gets_smaller}.
Recall that $M(r) = \frac{y}{x} = \frac{\lceil \frac{200}{\delta^2} \rceil + \lceil \frac{600 - 182\delta}{3\delta - 2\delta^2} \rceil}{\lceil \frac{200}{\delta^2} \rceil} 
\geq \frac{3 + \delta - 0.91\delta^2}{3 - 2\delta + 0.015\delta^2 - 0.01\delta^3}$. From Observation \ref{obs:get_coordinates_slope_of_line} we can obtain the x-coordinate of the intersection of $\ell(r)$ and our bound on $N''$ with $(\frac{63}{\delta^2}) \cdot (\frac{3 - 2\delta + 0.015\delta^2 - 0.01 \delta^3}{6 - \delta - 0.89\delta^2 - 0.01\delta^3}) \leq \frac{189 - 125.055\delta}{6\delta^2 - 1.905\delta^3}$ since $\delta$ is at most $0.1$. Similarly, we can obtain the x-coordinate of the intersection of the line with slope $\frac{3}{3-2\delta}$ and our bound  on $N''$ with ($\frac{63}{\delta^2}) \cdot (\frac{3-2\delta}{6-2\delta}) = \frac{189 - 126\delta}{6\delta^2 - 2\delta^3}$. If $\frac{189 - 126\delta}{6\delta^2 - 2\delta^3} - \frac{189 - 125.055\delta}{6\delta^2 - 1.905\delta^3}  > 3$ is true, then $M(p) > \frac{3}{3-2\delta}$ for all $p \in P$. We have $\frac{189 - 126\delta}{6\delta^2 - 2\delta^3} - \frac{189 - 125.055\delta}{6\delta^2 - 1.905\delta^3}  > 3$ which is true if $\frac{24.57 - 236.16\delta + 140.58\delta^2 + 22.86\delta^3}{72\delta - 46.86\delta^2 - 7.62\delta^3} > 0$. The denominator is always positive since $\delta$ is at most $0.1$. It follows from $\delta \leq 0.1$ that the numerator is always positive. Indeed, $24.57 - 236.16(0.1) + 140.58(0.1)^2 + 22.86(0.1)^3 = 2.38266$. Thus $M(p) > \frac{3}{3-2\delta}$ for all $p\in P$ which implies that for every $p \in Subtree(s)$ it must be that $M(p) > \frac{3}{3 - 2\delta}$.

Next we will show that $M(p) < 1 + \frac{4\delta}{3 - 2\delta} = \frac{3 + 2\delta}{3 - 2\delta}$ for every point $p$ that is $3$ above $\ell(t)$ such that $N'' \leq D(p) < N'$. Let $p \in P'$ be the set of all points such that $p$ is $3$ above $\ell(t)$ and $N'' \leq D(p) < N'$.
Recall that $M(t) = \frac{y + 21}{x - 21}= \frac{\lceil \frac{200}{\delta^2} \rceil + \lceil \frac{600 - 182\delta}{3\delta - 2\delta^2} \rceil + 21}{\lceil \frac{200}{\delta^2} \rceil -21} \leq \frac{3 + \delta }{3 - 2.5\delta}$. From Observation \ref{obs:get_coordinates_slope_of_line} we can obtain the x-coordinate of the intersection of $\ell(t)$ and our bound on $N''$ with $(\frac{63}{\delta^2}) \cdot (\frac{3-2.5\delta}{6 - 1.5\delta}) = \frac{189 - 157.5\delta}{6\delta^2 - 1.5\delta^3}$. Similarly, we can obtain the x-coordinate of the intersection of the line with slope $\frac{3+2\delta}{3-2\delta}$ and our bound on $N''$ with $(\frac{63}{\delta^2})\cdot (\frac{3-2\delta}{6}) = \frac{189 - 126\delta}{6\delta^2}$. If $\frac{189 - 157.5\delta}{6\delta^2 - 1.5\delta^3} - \frac{189 - 126\delta}{6\delta^2} > 3$ is true then $M(p) < \frac{3 + 2\delta}{3-2\delta}$ for all $p \in P'$. We then have $\frac{189 - 157.5\delta}{6\delta^2 - 1.5\delta^3} - \frac{189 - 126\delta}{6\delta^2} > 3$ is true if $\frac{186 - 594 \delta + 54\delta^2}{72\delta - 18\delta^2} > 0$. The denominator is always positive since $\delta$ is at most $0.1$. The numerator is positive since $\delta \leq 0.1$. Indeed, when $\delta = 0.1$ we have $186 - 594(0.1) + 54(0.1)^2 = 127.14$. Thus $M(p) < \frac{3 + 2\delta}{3-2\delta}$ for all $p \in P'$ which implies that for every $p \in Subtree(s)$ it must be that $ M(p) < \frac{3+2\delta}{3 - 2\delta}$. Therefore it must be that for every $p \in Subtree(s)$ that $1 + \frac{2\delta}{3-2\delta} < M(p) < 1 + \frac{4\delta}{3-2\delta}$. This concludes the proof of (1) from Lemma \ref{lem:lowerBound}.

\subparagraph*{Proving $s$ satisfies (2).}
We will first show that $N - D(s) > \frac{7+\delta}{\delta}$, and then use this bound to show that (2) must be true. Recall that $N = x + y = 2\lceil \frac{200}{\delta^2} \rceil + \lceil \frac{600 - 182\delta}{\delta(3-2\delta)} \rceil \geq \frac{1200 - 200 \delta - 182 \delta^2}{3\delta^2 - 2\delta^3}$.
 Next, recall that $D(s) < N' \leq \frac{382}{\delta^2 - \frac{2}{3}\delta^3}$.
 The difference between these two bounds is  $\frac{1200 - 200 \delta - 182 \delta^2}{3\delta^2 - 2\delta^3} - \frac{382}{\delta^2 - \frac{2}{3}\delta^3} = \frac{54 - 200\delta - 182\delta^2}{3\delta^2 - 2\delta ^3}$. 
Clearly we can see that $\frac{54 - 200\delta - 182\delta^2}{3\delta^2 - 2\delta ^3} > \frac{7 + \delta}{\delta}$ is true. It follows that $N - D(s) > \frac{7+\delta}{\delta}$. 
Let $d = D(s) + \frac{7+\delta}{\delta}$.
Recall that $s^\uparrow$ extends to at least one point in $\Sigma$ and $s^\rightarrow$ must also extend to at least one point in $\Sigma$. It then follows from $N - D(s) > \frac{7+\delta}{\delta}$ that $s^\uparrow$ extends to at least one point on $d$ and $s^\rightarrow$ must also extend to at least one point on $d$. We will show that this implies that $ConeWidth(Subtree(s),d) > 2 - \frac{2\delta}{3}$.

\begin{figure}
    \centering
    \includegraphics[scale=1]{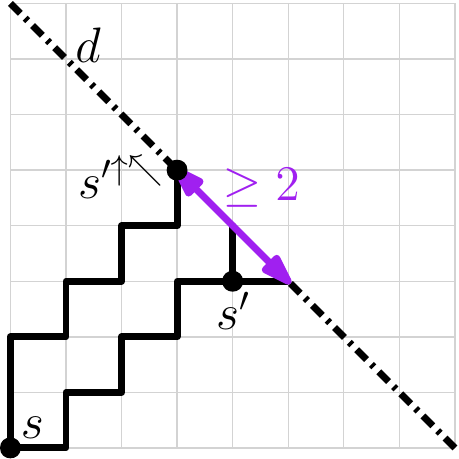}
    \caption{An illustration of the first case.}
    \label{fig:s_prime}
\end{figure}

There are two cases to consider.  The first is there is a split point in the path from $s^\rightarrow$ before $d$, and the second is there is no split point in the path from $s^\rightarrow$ before $d$.
Let us consider the first case. Let $s'$ be a split point in the path from $s^\rightarrow$. We will show that $ConeWidth(Subtree(s),D(s') + 1) \geq 2$.  By assumption $D(s') < d$ which implies that $D(s') + 1 \leq d$. Trivially, $s^\rightarrow \in R(s'^\rightarrow)$ and $s^\rightarrow \in R(s'^\uparrow)$. Furthermore, since $s^\uparrow$ must extend to $d$ it then trivially extends to $D(s') + 1$. We also know from (S3) that any point that $s^\uparrow$ extends to on $D(s') + 1$ must be above $s'^\uparrow$. Notice that the diagonal distance between $s'^\rightarrow$ and $s'^{\uparrow\nwarrow}$ is $2$. It then follows that $ConeWidth(Subtree(s),D(s')+1) \geq 2$ which implies $ConeWidth(Subtree(s),d) \geq 2$. See Figure \ref{fig:s_prime}.

\begin{figure}
    \centering
    \includegraphics[scale=1]{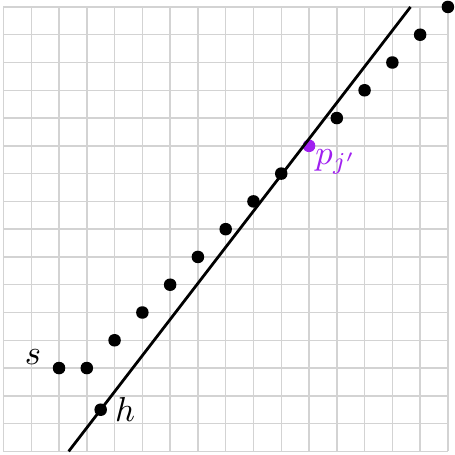}
    \caption{An illustration of the points of the form $p_j$ and the point $h$.}
    \label{fig:p_j_points}
\end{figure}

\begin{figure}
    \centering
    \includegraphics[scale=1]{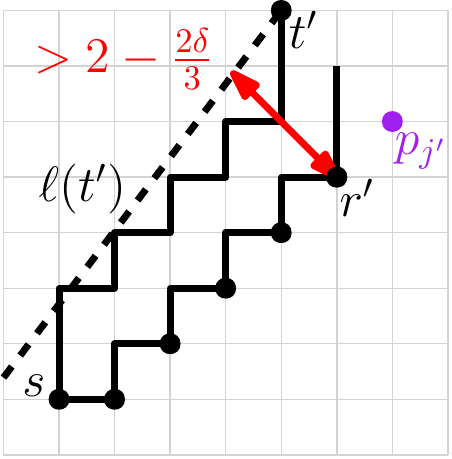}
    \caption{An illustration of how two consecutive movements on the path from $s^\rightarrow$ to $d$ implies that $ConeWidth(Subtree(s),D(r') > 2 - \frac{2\delta}{3}$.}
    \label{fig:two_consecutive_vert_movements}
\end{figure}

In the second case we consider the point $h := (s_x + 1.5, s_y - 1.5)$. Recall that the error is less than $1.5$.  This implies that the bottom edge of $Cone(Subtree(s))$ must intersect $D(s)$ above $h$. See Figure \ref{fig:p_j_points}. 
It follows that all points in $Subtree(s^\rightarrow)$ have slope greater than $M(h)$ otherwise the error is at least $1.5$. Let the points $p_j$ be of the form $p_j := (s^\rightarrow_x + j, s^\rightarrow_y + j)$ for any $j \geq 1$.  If $s^\rightarrow$ extends to a $p_j$ then it must be that the path from $s^\rightarrow$ to diagonal $D(p_j)$ has $j$ vertical and $j$ horizontal movements. We would like to know the maximum $j$ for which $s^\rightarrow$ can extend to $p_j$ such that $M(p_j) > M(h)$.  Notice that $s^\rightarrow = (h_x - 0.5, h_y + 1.5)$ which implies $M(p_j) = \frac{h_y + 1.5 + j}{h_x - 0.5 + j}$.  We will now solve for which $j$ satisfies $\frac{h_y + 1.5 + j}{h_x - 0.5 + j} \leq M(h)$. Notice this is equivalent as asking for what $j$ satisfies $\frac{h_y + 1.5 + j}{h_x - 0.5 + j} \leq \frac{h_y}{h_x}$. We then have $h_yh_x + h_x(1.5 + j) \leq h_yh_x + h_y(j-0.5)  \implies \frac{1.5 + j}{j - 0.5} \leq \frac{h_y}{h_x}$. Thus we would like to know for which $j$ satisfies $\frac{1.5 + j}{j - 0.5} \leq M(h)$. 

Given a fixed $M(s)$, we can see that $M(h)$ grows as $D(s)$ becomes larger.  Recall that $M(s) > \frac{3}{3-2\delta}$ and $D(s) \geq N''$.
 It follows that $M(h)$ is at least the slope of the point, $q$, that is $1.5$ diagonal distance below  $Inter(\ell(3-2\delta,3), N'')$. From Observation \ref{obs:get_coordinates_slope_of_line} we then have $q_x = \frac{3-2\delta}{6-2\delta}N'' + 1.5 \leq \frac{189 - 117\delta}{6\delta^2 - 2\delta^3}$ and $q_y = N'' - \frac{189 - 117\delta}{6\delta^2 - 2\delta^3} = \frac{189 - 9\delta}{6\delta^2 - 2\delta^3}$. We can then conclude that $M(h) \geq \frac{189 - 9\delta}{189 - 117\delta}$.
 
 It follows that if we find for which $j$ is $\frac{1.5 + j}{j - 0.5} \leq \frac{189 - 9\delta}{189 - 117\delta}$ true, it will also be true for $\frac{1.5 + j}{j - 0.5} \leq  M(h)$. We then have
 $\frac{1.5 + j}{j - 0.5} \leq \frac{189 - 9\delta}{189 - 117\delta}$ is true for all $j \geq \frac{378 - 180\delta}{108\delta}$. Therefore the maximum $j$ for which $s^\rightarrow$ can extend to $p_j$ is $j = \lceil \frac{378 - 180\delta}{108\delta}\rceil - 1$. Let $j' := \lceil \frac{378 - 180\delta}{108\delta}\rceil$, that is the maximum $j$ plus $1$. Let us then consider what point $s^\rightarrow$ must extend to on $D(p_{j'})$. Because $M(p_{j' }) < M(h)$  it must be that $s^\rightarrow$ extends to a point who's slope is at minimum $M(p_{j'}^\nwarrow)$. Recall that $p_{j'}$ is $j'$ vertical and $j'$ horizontal movements away from $s^\rightarrow$. It then follows that $p_{j'}^\nwarrow$ is $j'+1$ vertical and $j'-1$ horizontal movements away from $s^\rightarrow$. Therefore the path from $s^\rightarrow$ to $D(p_{j'})$ must have at least two more vertical than horizontal movements. It then follows from pigeonhole principle that the path from $s^\rightarrow$ to $D(p_{j'})$ must contain at least two consecutive vertical movements. See Figure \ref{fig:two_consecutive_vert_movements}.

Notice that $D(p_{j'}) = D(s) + 2j' + 1 = D(s) + 2\cdot \lceil \frac{378 - 180\delta}{108\delta} \rceil + 1 < D(s) + 2 \cdot (\frac{378 - 180\delta}{108\delta} + 1 ) + 1 \leq D(s) + \frac{7}{\delta}$. 
Recall that $d = D(s) + \frac{7+ \delta}{\delta}$. Clearly $D(p_{j'}) < d$. Therefore the path from $s^\rightarrow$ to $d$ must contain at least two consecutive vertical movements. We will now show these consecutive vertical movements imply $ConeWidth(Subtree(s),d) > 2 - \frac{2\delta}{3}$.
Let $r'$ be a point such that $D(r') \leq d-2$ and $s^\rightarrow$ extends to $r'$, $r'^\uparrow$, and $r'^{\uparrow\uparrow}$. Let $t' := r'^{\uparrow\uparrow\nwarrow}$. Recall that $s^\uparrow$ must extend to $N$ which implies $s^\uparrow$ also extends to  $D(t')$. We also know from (S3) that any point that $s^\uparrow$ extends to on $D(t')$ must be above $r'^{\uparrow\uparrow}$.  This implies that $ConeWidth(Subtree(s),D(r'))$ is at least the diagonal distance between $r'$ and $Inter(\ell(t'), D(r'))$. 
As we showed in our proof of (1) from Lemma \ref{lem:lowerBound} all points $s$ extends to must have slope less than $\frac{3 + 2\delta}{3 - 2\delta}$ which implies that $M(t') < \frac{3 + 2\delta}{3 - 2\delta}$, otherwise there is no point $s^\uparrow$ can extend to on $D(t')$. Let $z$ be the point that is $2-\frac{2\delta}{3}$ above $r'$.
Let $\overline{t'}$ be the Euclidean line that intersects both $t'$ and $z$. 
If $Inter(\ell(t'), D(r'))$ is above $Inter(\overline{t'}, D(r'))$ then $ConeWidth(Subtree(s),D(r')) > 2 - \frac{2\delta}{3}$. 
We then have $M(\overline{t'}) = \frac{t'_y - z_y }{t'_x - z_x}  =  \frac{t'_y - (t'_y - (1 + \frac{2\delta}{3}))}{t'_x - (t'_x -(1-\frac{2\delta}{3}))} = \frac{3 + 2\delta}{3 - 2\delta}$. It then follows that $M(t') < M(\overline{t'})$ which implies that $Inter(\ell(t'), D(r'))$ is above $Inter(\overline{t'}, D(r'))$. Thus $ConeWidth(Subtree(s),D(r')) > 2 - \frac{2\delta}{3}$ which implies that $ConeWidth(Subtree(s),d) > 2 - \frac{2\delta}{3}$. This concludes the proof of (2) from Lemma \ref{lem:lowerBound}, which then concludes the entire proof of Lemma \ref{lem:lowerBound}.

\end{proof}

Using the existence of $s$ from Lemma \ref{lem:lowerBound}, we can prove Theorem \ref{thm:lowerBound}.  We show that given such an $s$, there there must be a point $p$ as described in the proof sketch.  Then it is the case that $p$ is on a segment whose error is at least $1.5 - \delta > 1.5 - \epsilon$.

\begin{proof}
Let $s$ be a split point as described in Lemma \ref{lem:lowerBound}. We will show that there must be a point $p$ such that $p \in R(s)$ and the parent of $p$ will be more than $1.5 - \delta$ from an edge of $Cone(Subtree(s))$.

\begin{figure}
    \centering
    \includegraphics{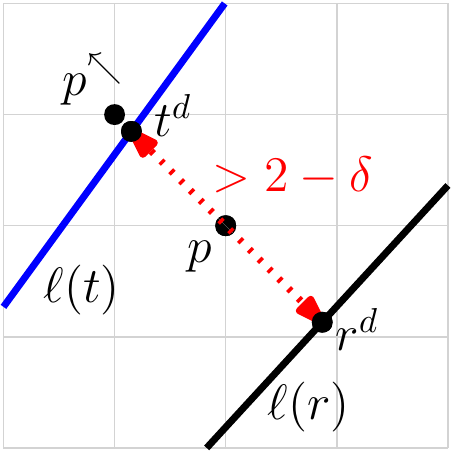}
    \caption{$d \in \mathcal{D}$ that satisfies the properties below. The red dotted line represents $ConeWidth(Subtree(s),d) > 2-\delta$, and $t^d$ is at most $\frac{2}{3}\delta$ below $p^\nwarrow$.}
    \label{fig:diagonal_d_in_D}
\end{figure}

Let $\ell(t)$ be the left edge of $Cone(Subtree(s))$ and $\ell(r)$ be the right edge of $Cone(Subtree(s))$. For a diagonal $d$ let $t^d$ be $Inter(\ell(t),d)$, and let $r^d$ be $Inter(\ell(r),d)$.
Let $\mathcal{D}$ be the set of all diagonals such that $d \in \mathcal{D}$
if $d \leq D(s)$, $0 \leq t_x^d - \lfloor t_x^d \rfloor \leq \frac{2}{3}\delta$, and $ConeWidth(Subtree(s), d) > 2 - \delta$. That is to say $d\in \mathcal{D}$ if $d \leq D(s)$, $t^d$ is at most $\frac{2}{3}\delta$ below a grid point on $d$, and the cone width of $s$ on $d$ is more than $2-\delta$ (Figure \ref{fig:diagonal_d_in_D}).  Let $d \in \mathcal{D}$ be such a diagonal.
Let $p := (\lceil t^d_x  \rceil, \lfloor t^d_y \rfloor )$. Notice that all points above $p$ on $D(p)$ are more than $1.5 - \delta$ away from $r^d$ and all points below $p$ on $D(p)$ are more than $1.5 - \delta$ away from $t^d$. It then follows that $p \in R(s)$ and $r^d_x - p_x < 1.5 - \delta$, otherwise the error is more than $1.5 - \delta$. 

There are two cases to consider. The first is if $p^\leftarrow$ is the parent of $p$, and the second is if $p^\downarrow$ is the parent of $p$.

\begin{figure}
    \centering
    \includegraphics{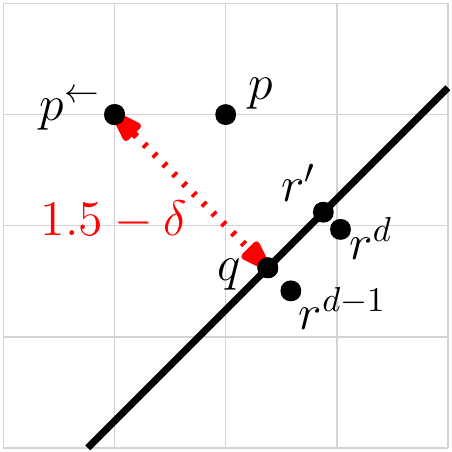}
    \caption{The case when $p^\leftarrow$ is the parent of $p$. The red dotted line represents the distance between $p^\leftarrow$ and $q$.}
    \label{fig:p_left_arrow_parent}
\end{figure}

Let us first consider when $p^\leftarrow$ is the parent of $p$. We will show that $r^{d-1}_x - p^\leftarrow_x > 1.5 - \delta$. Let $r'$ be the point that is exactly $\delta$ above $p^\searrow$, that is $r' = (p_x + 1 - \delta, p_y - 1 + \delta)$. Note that $r^d$ will always be strictly below $r'$.
Let $q$ be the point that is $1.5-\delta$ below $p^\leftarrow$, that is $q := (p_x + 0.5 - \delta, p_y - 1.5 + \delta)$ (Figure \ref{fig:p_left_arrow_parent}). Note that $r^{d-1}$ must be above or on $q$ in order for the error to be at most $1.5 - \delta$. The slope of the Euclidean line that intersects both $q$ and $r'$ is $\frac{r'_y - q_y}{r'_x - q_x} = \frac{p_y - 1 + \delta - (p_y - 1.5 + \delta)}{p_x + 1 - \delta - (p_x + 0.5 - \delta)} = 1$. This implies that the slope of $\ell(r)$ would need to be less than $1$ in order for $r^{d-1}$ to be above or on $q$. However,  
Lemma \ref{lem:lowerBound} states that $M(\ell(r)) > \frac{3}{3-2\delta}> 1$. It follows that $r^{d-1}$ is below $q$ and therefore $r^{d-1}_x - p^\leftarrow_x > 1.5 - \delta$.

\begin{figure}
    \centering
    \includegraphics{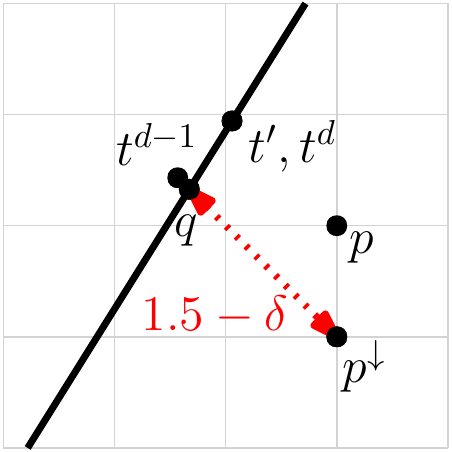}
    \caption{The case when $p^\downarrow$ is the parent of $p$. The red dotted line represents the distance between $p^\downarrow$ and $q$.}
    \label{fig:p_down_arrow_parent}
\end{figure}

Next we consider the case when $p^\downarrow$ is the parent of $p$. We will show that $p^\downarrow_x - t^{d-1}_x > 1.5 - \delta$. Let $t'$ be the point that is $\frac{2}{3}\delta$ below $p^\nwarrow$. Note that $t^d$ will always be strictly above $t'$. Let $q$ be the point that is $1.5 - \delta$ above $p^\downarrow$, that is $q :=(p_x - 1.5 + \delta, p_y + 0.5 - \delta)$ (Figure \ref{fig:p_down_arrow_parent}). Note that $t^{d-1}$ must be below or on $q$ in order for the error to be at most $1.5 - \delta$. The slope of the Euclidean line that intersects both $q$ and $t'$ is $\frac{p_y + 1 - \frac{2}{3}\delta - (p_y + 0.5 - \delta)}{p_x - 1 + \frac{2}{3}\delta - (p_x - 1.5 + \delta)} = \frac{3 + 2\delta}{3 - 2\delta}$. This implies the slope of $\ell(t)$ would need to be more than $\frac{3+2\delta}{3-2\delta}$ in order for $t^{d-1}$ to be below or on $q$. However, Lemma \ref{lem:lowerBound} states that $M(\ell(t)) < \frac{3 + 2\delta}{3 - 2\delta}$. It follows that $t^{d-1}$ is above $q$ and therefore $p^\downarrow_x - t^{d-1}_x > 1.5 - \delta$.

Therefore if we show that at least one $d \in \mathcal{D}$ exists then we are done.  
We will do so by first showing there exists a diagonal $d^*$
that is the smallest diagonal for which the cone width of $Subtree(s)$ is greater than $2-\delta$.
We will then show that there must be a diagonal greater than $d^*$ and less than or equal to $D(s)$ for which $\ell(t)$ crosses at most $\frac{2}{3}\delta$ below a grid point.

Let us begin by bounding for which diagonals the cone width of $Subtree(s)$ must be greater than $2 - \delta$. Recall Lemma \ref{lem:lowerBound} states that $ConeWidth(Subtree(s),D(s)+ \frac{7+\delta}{\delta}) > 2-\frac{2}{3}\delta$. 
Trivially, because $\ell(t)$ and $\ell(r)$ are Euclidean lines that intersect the origin the diagonal distance between them is directly proportional to the diagonal in which we are measuring the distance. It then follows that $ConeWidth(Subtree(s),D(s) - d') = (\frac{D(s) - d'}{D(s) + \frac{7 + \delta}{\delta}}) \cdot (ConeWidth(Subtree(s), D(s) + \frac{7 + \delta}{\delta})) > (\frac{D(s) - d'}{D(s) + \frac{7 + \delta}{\delta}}) \cdot (2 - \frac{2\delta}{3})$ for some diagonal $D(s) - d'$. Let us now solve for a bound on $d'$ for when $ConeWidth(Subtree(s), D(s) - d') > 2- \delta$. We have $(\frac{D(s) - d'}{D(s) + \frac{7 + \delta}{\delta}}) \cdot (2 - \frac{2}{3}\delta) > 2 - \delta \implies d' < \frac{\delta^2 D(s) + 15\delta + 3 \delta^2 - 42}{6\delta - 2\delta^2}$. Let $\hat{d} := \frac{\delta^2 D(s) + 15\delta + 3 \delta^2 - 42}{6\delta - 2\delta^2} - 1 = \frac{\delta^2 D(s) + 9\delta + 5 \delta^2 - 42}{6\delta - 2\delta^2} $. Notice that the denominator of $\hat{d}$ is positive since $\delta \leq 0.1$ by definition of $s$,
and the numerator of $\hat{d}$ is positive because we know from Lemma \ref{lem:lowerBound} that $D(s) \geq \frac{63}{\delta^2}$, indeed the numerator is at least  $\delta^2 \frac{63}{\delta^2} + 9\delta + 5 \delta^2 - 42 = 21 + 9\delta + 5 \delta^2$. This implies that $D(s) - \hat{d} < D(s)$, and that $d^* = \lceil D(s) - \frac{\delta^2 D(s) + 15\delta + 3 \delta^2 - 42}{6\delta - 2\delta^2} \rceil < D(s)$.  This implies that $ConeWidth(Subtree(s), D(s)) > 2- \delta$.

\begin{figure}
    \centering
    \includegraphics{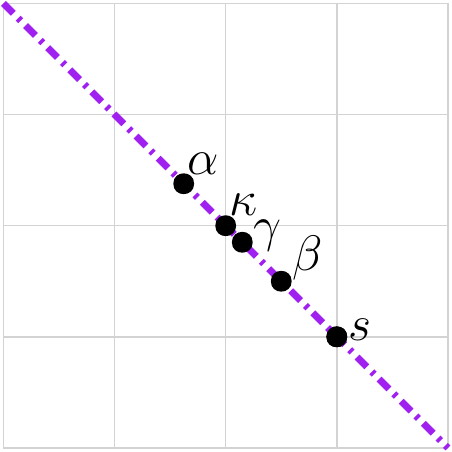}

    \caption{The cases for where $\ell(t)$ crosses $D(s)$ as described above. $D(s)$ is represented by the purple dashed dotted line. }
    
    \label{fig:p_parent_cases}
\end{figure}

We will now show that $\ell(t)$ crosses a diagonal between $[d^*,D(s)]$ at most $\frac{2}{3}\delta$ below a grid point. 
We begin by bounding where $\ell(t)$ can cross $D(s)$.
Let $\alpha$ be the point that is $1.5 - \delta$ above $s$, let $\kappa := s^\nwarrow$, let $\gamma$ be the point that is $\frac{2}{3}\delta$ below $\kappa$, and let $\beta$ be the point that is $0.5$ above $s$ (Figure \ref{fig:p_parent_cases}).
Trivially, $t^{D(s)}$ must be below or on $\alpha$, otherwise $t^{D(s)}$ is more than $1.5 -\delta$ above $s$. It follows from $ConeWidth(Subtree(s), D(s)) > 2 - \delta$ that $t^{D(s)}$ is above or on $\beta$, otherwise $r^{D(s)}$ is more than $1.5 - \delta$ below $s$.  Therefore it must be that $\alpha_x \leq  t_x^{D(s)} \leq \beta_x$. 
There are three cases to consider. The first is that $\kappa_x \leq t_x^{D(s)} \leq \gamma_x$. Notice for this case that $t^{D(s)}$ is at most $\frac{2}{3}\delta$ below the grid point $\kappa$.  Since $d^* < D(s)$
it follows that $D(s) \in \mathcal{D}$ and $p = s$ for this case and we are done. 
The second case is that $\gamma_x < t_x^{D(s)}  \leq \beta_x$, and the third case is that $\alpha_x \leq  t_x^{D(s)} < \kappa_x$. Recall by definition of $s$ that  $\delta \leq 0.1$ which implies $\gamma_x < \beta_x$ for the second case and $\alpha_x < \kappa_x$ for the third case.

We now analyze the second case when $\gamma_x < t_x^{D(s)}  \leq \beta_x$. The third case when $\alpha_x \leq  t_x^{D(s)} < \kappa_x$ is analyzed the same as the second and we make note of any differences between the two cases.

Let the grid points $q_j$ be of the form $q_j := (s_x^\nwarrow + 1 - j, s^\nwarrow_y - j)$ for $j \geq 0$ ($q_j := (s_x^\nwarrow - j, s^\nwarrow_y - j)$ for $j \geq 0$ for the third case). 
Let $Q$ be the set of all $q_j$ points and let $\ell(Q)$ be the Euclidean line that intersects all points in $Q$ (Figure \ref{fig:q_j_points}). For any diagonal $d$ let $Q^d$ be the point where $\ell(Q)$ and $d$ intersect. Let $v$ be the intersection of $\ell(t)$ and $\ell(Q)$. Notice that $D(v) \leq D(s)$ since $t_x^{D(s)} \leq \beta_x$, $Q^{D(s)}  = \beta$, $M(\ell(t)) > 1$, and $M(\ell(Q)) = 1$ ($D(v) < D(s)$ since $t^{D(s)}_x < s^\nwarrow_x$ and $Q^{D(s)} = s^\nwarrow$ for the third case). 
We will show that a $q^* \in Q$ exists such that $0 \leq t^{D(q^*)}_x - q^*_x \leq \frac{2}{3}\delta$ and $D(q^*) \in \mathcal{D}$.

\begin{figure}[ht]
\centering
\begin{tabular}{c@{\hspace{0.1\linewidth}}c}

\includegraphics[scale=1]{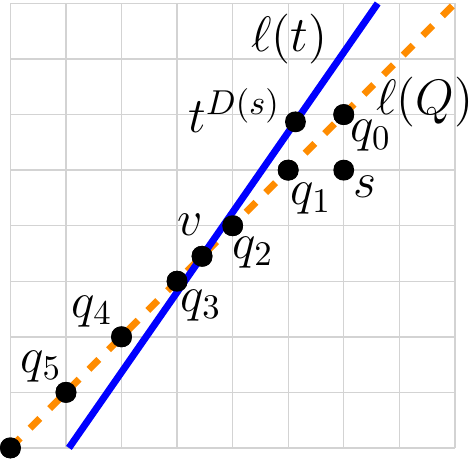}&
\includegraphics[scale=1]{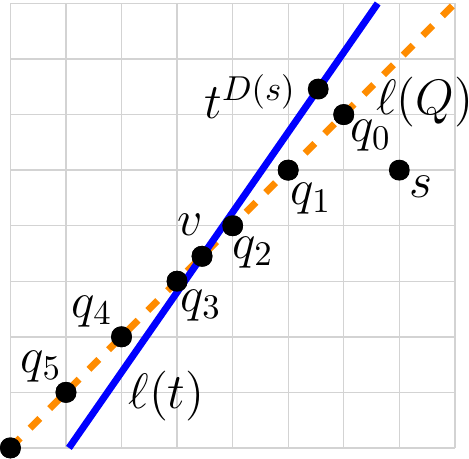}

\\

 (a) & (b)
\end{tabular}
 \caption{ (a) Second case illustration. (b) Third case illustration. 
 }
 \label{fig:q_j_points}
\end{figure}

We will do this by first showing that there must be a $q \in Q$ such that $0 \leq t^{D(q)}_x - q_x \leq \frac{2}{3}\delta$ and $D(v) - D(q) < 2$.

Let $v' := (\lceil v_x \rceil, \lceil v_y \rceil)$ and $v'' := v'^\swarrow$ (Figure \ref{fig:v_prime_and_w}). Notice that $v', v'' \in Q$. There are two cases. Let $\mathcal{A}$ denote the first case for when $v = v'$, and let $\mathcal{B}$ denote the second case for when $v \neq v'$. Trivially $q = v'$ for case $\mathcal{A}$. For case $\mathcal{B}$ let the point $w$ be the point $\frac{2}{3}\delta$ below $v''$. Notice that $w = (v'_x - 1 + \frac{2}{3}\delta, v'_y - 1 - \frac{2}{3}\delta)$. The slope of the line that intersects both $v'$ and $w$ is $\frac{v'_y - w_y}{v'_x - w_x} = \frac{v'_y - (v'_y - 1 - \frac{2}{3}\delta)}{v'_x - (v'_x - 1 + \frac{2}{3}\delta)} = \frac{3 + 2\delta}{3 - 2\delta}$. From Lemma \ref{lem:lowerBound} we know that $M(\ell(t)) < \frac{3 + 2\delta}{3- 2\delta}$. It follows that $0 \leq t^{D(v'')}_x - v''_x \leq \frac{2}{3}\delta$.

Trivially, $D(v') - D(v'') = 2$ which implies $D(v) - D(v'') < 2$. Which implies that $q = v''$ in this case. Therefore there must be a $q \in Q$ such that $0 \leq t^{D(q)}_x - q_x \leq \frac{2}{3}\delta$ and $D(v) - D(q) < 2$ for either case $\mathcal{A}$ or $\mathcal{B}$.

\begin{figure}
    \centering
    \includegraphics{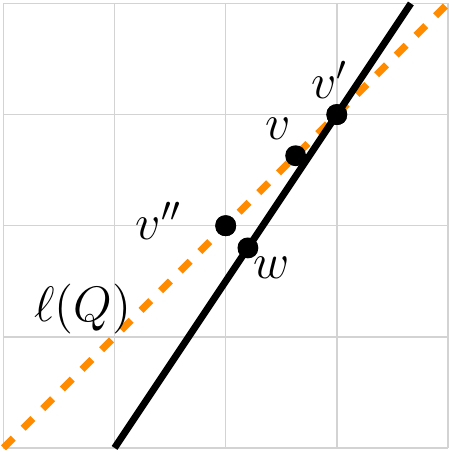}

    \caption{An illustration of the Euclidean line that passes through $v'$ and $w$. $\ell(Q)$ is represented by an orange dashed line.}
    
    \label{fig:v_prime_and_w}
\end{figure}

Next we will show that $d^* < D(q)$ by showing that $d^* < D(v) - 2$. 
Let $q_{j'} := (s^\nwarrow_x + 1 - j', s^\nwarrow_y - j')$ for $j' := \frac{3}{2\delta}$ ($q_{j'} := (s^\nwarrow_x - j', s^\nwarrow_y - j')$ for $j' := \frac{3 - \delta}{2\delta}$ for the third case). Note that while $q_{j'} \in \ell(Q)$, $j'$ might not be an integer which implies $q_{j'}$ might not be in $Q$ (Figure \ref{fig:q_j_prime}). The slope of the line that intersects both $s^\nwarrow$ and $q_{j'}$ is $\frac{s^\nwarrow_y - (s^\nwarrow_y - j')}{s^\nwarrow_x - (s^\nwarrow_x + 1 - j')} = \frac{s^\nwarrow_y - (s^\nwarrow_y - \frac{3}{2\delta})}{s^\nwarrow_x - (s^\nwarrow_x + 1 - \frac{3}{2\delta})} = \frac{3}{3-2\delta}$ (slope of the line the intersects the point $0.5$ above $s^\nwarrow$ and $q_{j'}$ is $\frac{s^\nwarrow_y + 0.5 - (s^\nwarrow_y - \frac{3-\delta}{2\delta})}{s^\nwarrow_x - 0.5 - (s^\nwarrow_x - \frac{3-\delta}{2\delta})} = \frac{3}{3-2\delta}$ for the third case). From Lemma \ref{lem:lowerBound} we know that $M(\ell(t)) > \frac{3}{3-2\delta}$, and recall that $t_x^{D(s)}  > \gamma_x$ ($t_x^{D(s)} \geq \alpha_x$ for the third case). It then follows that $D(q_{j'}) < D(v)$ which implies that $D(q_{j'}) - 2< D(v) - 2$. Note that $s^\nwarrow_x + s^\nwarrow_y = D(s)$. We have $D(q_{j'}) - 2 = s^\nwarrow_y - \frac{3}{2\delta} + s^\nwarrow_x + 1 - \frac{3}{2\delta} - 2 = D(s) - \frac{3 + \delta}{\delta}$ ($D(q_{j'}) - 2 = s^\nwarrow_y - \frac{3 - \delta}{2\delta} + s^\nwarrow_x - \frac{3 - \delta}{2\delta} - 2 = D(s) - \frac{3 + \delta}{\delta}$ for the third case). We then would like to know if $D(q_{j'}) - 2 > d^*$. That is $D(s) - \frac{3 + \delta}{\delta} > \lceil D(s) - \frac{\delta^2 D(s) + 15\delta + 3 \delta^2 - 42}{6\delta - 2\delta^2} \rceil$. Note that $\lceil D(s) - \frac{\delta^2 D(s) + 15\delta + 3 \delta^2 - 42}{6\delta - 2\delta^2} \rceil <  D(s) - \frac{\delta^2 D(s) + 15\delta + 3 \delta^2 - 42}{6\delta - 2\delta^2} +1$. So we can solve if $D(s) - \frac{3 + \delta}{\delta} > D(s) - \frac{\delta^2 D(s) + 15\delta + 3 \delta^2 - 42}{6\delta - 2\delta^2} +1$ is true. We have $D(s) - \frac{3 + \delta}{\delta} - ( D(s) - \frac{\delta^2 D(s) + 15\delta + 3 \delta^2 - 42}{6\delta - 2\delta^2} +1) > 0 \implies \frac{\delta^2D(s) + 9 \delta + 7\delta^2 - 60}{6\delta - 2\delta^2} > 0$. Notice that the denominator is positive since $\delta \leq 0.1$ by definition of $s$, and the numerator is positive because we know from Lemma \ref{lem:lowerBound} that $D(s) > \frac{63}{\delta^2}$, indeed the numerator is at least $\delta^2\frac{63}{\delta^2} + 9\delta + 7\delta^2 - 60 = 63 + 9\delta + 7\delta^2 - 60 = 3 + 9\delta + 7\delta^2$. It follows that $d^* < D(q_{j'}) - 2$ which implies that $d^* < D(v) - 2$. It then follows from $D(v) - D(q) < 2$ that $d^* < D(q)$. We can then conclude that there exists a $q^*$ as described above for all three cases.

\begin{figure}
    \centering
    \includegraphics{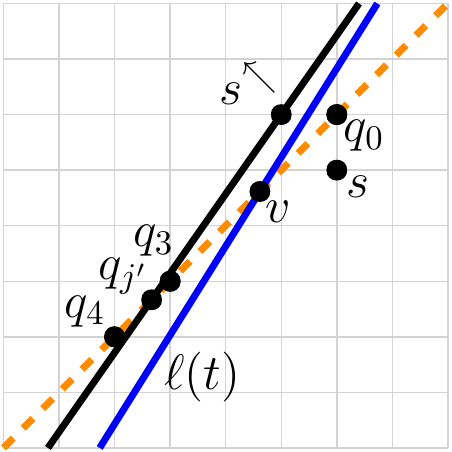}

    \caption{An illustration of $q_{j'}$. $\ell(Q)$ is represented by an orange dashed line.}
    
    \label{fig:q_j_prime}
\end{figure}

It then follows from the definition of $\mathcal{D}$ that there must also be a point $p \in R(s)$ whose parent is more than $1.5 -\delta$ from an edge of $Cone(Subtree(s))$. Therefore the error is greater than $1.5-\delta$.

\end{proof}

\bibliography{bibliography}

\end{document}